\documentclass[runningheads,english,numberwithinsect,envcountsame]{llncs}	
	
\usepackage[T1]{fontenc}
\usepackage[english]{babel}
\usepackage[utf8]{inputenc}
\usepackage{amssymb,amsmath}
\usepackage{graphicx}
\usepackage{paralist}
\usepackage{hyperref}
\usepackage{color}
\usepackage{xspace}

\usepackage{subfig}
\usepackage{wrapfig}
\usepackage{chngcntr}
\usepackage{apptools}
\usepackage[export]{adjustbox} %vertical alignment of figures
\usepackage[pagewise]{lineno}
\usepackage{pifont}
\usepackage{multirow}
\usepackage[symbol]{footmisc}

\usepackage{diagbox}% cell in a table with diagonal line and two values

\AtAppendix{\counterwithin{lemma}{section}}
\AtAppendix{\counterwithin{theorem}{section}}

\usepackage{marcel_notation}

\newif\ifshort

\shorttrue

\newenvironment{claimproof}{}{\hfill$\triangleleft$}

\title{Aligned Drawings of Planar Graphs \thanks{Partially supported by grant WA
	654/21-1 of the German Research Foundation (DFG).  A preliminary version of
	this paper appeared in the {\em Proceedings of the 25th International
		Symposium on Graph Drawing (GD '17)}. }}

\author{Tamara Mchedlidze\inst{1} \and Marcel Radermacher\inst{1} \and Ignaz
Rutter\inst{2} \and}

\institute{
		Department of Computer Science, Karlsruhe  Institute of Technology,
	Germany \and
	Department of Computer Science and Mathematics, University
	of Passau, Germany \\
	\email{mched@iti.uka.de}, \email{radermacher@kit.edu}, \email{rutter@fim.uni-passau.de}}

\begin{document}
\maketitle

\begin{abstract} 
	Let $G$ be a graph that is topologically embedded in the plane and let $\CurArr$ be
	an arrangement of pseudolines intersecting the drawing of $G$.  An
	\emph{aligned} drawing of $G$ and $\CurArr$ is a planar polyline drawing
	$\Gamma$ of $G$ with an arrangement $\Arr$ of lines so that $\Gamma$ and
	$\Arr$ are homeomorphic to $G$ and $\CurArr$.  We show that if	$\CurArr$ is
	stretchable and every edge $e$ either entirely lies on a pseudoline or it has
	at most one intersection with $\CurArr$, then $G$ and $\CurArr$ have a
	straight-line aligned drawing.  In order to prove this result, we
	strengthen a result of Da Lozzo et al.~\cite{DaLozzo2016}, and prove that
	a planar graph $G$ and a single pseudoline $\Cur$ have an aligned drawing
	with a prescribed convex drawing of the outer face.	We also study the less
	restrictive version of the alignment problem with respect to one line,
	where only a set of vertices is given and we
	need to determine whether they can be collinear. We show that the problem
	is $\cNP$-complete but fixed-parameter tractable.  \end{abstract}

\section{Introduction}
\label{sec:introduction}

Two fundamental primitives for highlighting structural properties of a
graph in a drawing are \emph{alignment} of vertices such that they are
collinear, and geometric \emph{separation} of unrelated graph parts, e.g.,
by a straight line.  Both
these techniques have been previously considered from a theoretical
point of view in the case of planar straight-line drawings.

Da Lozzo et al.~\cite{DaLozzo2016} study the problem of producing a planar
straight-line drawing of a given embedded graph $G=(V,E)$ (i.e., $G$ has a fixed
combinatorial embedding and a fixed outer face) such that a given set $S \subseteq V$ of
vertices is collinear.  It is clear that if such a drawing exists, then the line
containing the vertices in $S$ is a simple curve starting and ending at infinity that
for each edge $e$ of $G$ either fully contains $e$ or intersects $e$ in at most
one point, which may be an endpoint.  We call such a curve a \emph{pseudoline
with respect to $G$}.  
%Note that the pseudoline contains all the vertices in $S$
%and it traverses the outer face.  
Da Lozzo et al.~\cite{DaLozzo2016} show that
this is a full characterization of the alignment problem, i.e., a planar straight-line
drawing where the vertices in $S$ are collinear exists if and only if there
exists a pseudoline $\mathcal L$ with respect to $G$ that contains the vertices
in $S$. However, the computational complexity of deciding whether such a
pseudoline exists is an open problem, which we consider in this paper.

Likewise, for the problem of separation, Biedl et al.~\cite{Biedl1998}
considered so-called $HH$-drawings where, given an embedded graph $G=(V,E)$ and
a partition $V = A \cupdot B$, one seeks a $y$-monotone planar polyline drawing
of $G$ with few bends in which $A$ and $B$ can be separated by a line.  Again,
it turns out that such a drawing exists if there exists a pseudoline
$\mathcal L$ with respect to $G$ such that the vertices in $A$ and $B$ are
separated by $\mathcal L$.  As a side-result Cano et
al.~\cite{doi:10.1137/130924172} extend the result of Biedl et al. to planar
straight-line drawings with a given star-shaped outer face.
% defined by $\mathcal L$.

The aforementioned results of Da Lozzo et al.~\cite{DaLozzo2016} show that given
a pseudoline $\mathcal L$ with respect to $G$ one can always find a planar
straight-line drawing of $G$ such that the vertices on $\mathcal L$ are
collinear and the vertices contained in the half-planes defined by $\mathcal L$
are separated by a line $L$.  In other words, a topological configuration
consisting of a planar embedded graph~$G$ and a pseudoline with respect to $G$
can always be stretched.  In this paper, we initiate the study of this
stretchability problem with more than one given pseudoline.

\begin{figure}[tb]
	\centering
	\subfloat[]{
		\includegraphics[page=5]{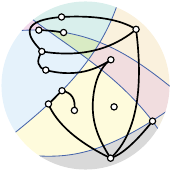}
	}
	\quad
	\subfloat[]{
		\includegraphics[page=6]{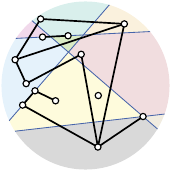}
	}
	\quad
	\subfloat[\label{fig:intro:aligned_drawing:pappus}]{
		\includegraphics[page=4]{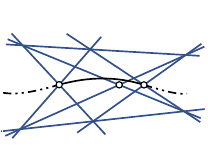}
	}
	\caption{(Pseudo-) Lines are depicted as blue curves, edges are black. The
		color of the cells indicates the bijection $\phi$ between the cells of
		$\CurArr$ and $\Arr$.
		Aligned drawing (b) of a 2-aligned planar embedded graph~(a).  (c) A
		non-stretchable arrangement of $9$
pseudolines (blue and black), which can be seen as a stretchable arrangement of $8$ pseudolines
(blue) and an edge (black solid). }
	\label{fig:intro:aligned_drawing}
\end{figure}

More formally, a pair $(G, \CurArr)$ is a \emph{$k$-aligned graph} if $G=(V, E)$
is a planar embedded graph and $\CurArr = \{\Cur_1, \dots, \Cur_k\}$ is an
arrangement of (pairwise intersecting) pseudolines with respect to $G$. In case
that every pair of distinct pseudolines intersect at most once, we refer to
$\CurArr$ as a \emph{pseudoline arrangement}.  If the number $k$ of pseudolines is
clear from the context, we drop it from the notation and simply speak of
\emph{aligned graphs}. For $1$-aligned graphs we write $(G, \Cur)$ instead of
$(G,\{\Cur\})$.  Let $A = \{L_1, \dots, L_k\}$ be a line arrangement   and
$\Gamma$ be a planar drawing of $G$.
A tuple $(\Gamma,A)$ is an \emph{aligned drawing of $(G, \CurArr)$} if and only
if the arrangement of the union of $\Gamma$ and $A$ is homeomorphic to the
arrangement of the union of $G$ and $\CurArr$.
A (pseudo)-line arrangement divides the plane into a set of \emph{cells}
$\Cell_1, \Cell_2, \dots, \Cell_\ell$. If $A$ is homeomorphic to $\CurArr$, then
there is a bijection $\phi$ between the cells of $\CurArr$ and the cells of $\Arr$.
If  $(\Gamma, A)$ is an aligned drawing of $(G, \CurArr$), then it has
the following properties; refer to Fig.~\ref{fig:intro:aligned_drawing}(a-b).
\begin{inparaenum}[(i)]
%\begin{compactenum}[(i)]
	%\iten the enclosed regions of \mathcal{L} are homeomorphic to the enclosed
	%regions of $\A$,\todo{A is stretchable to $L$.
	\item The arrangement of $A$ is homeomorphic to the arrangement of $\CurArr$
		(i.e., $\CurArr$ is \emph{stretchable} to $A$),
	\item $\Gamma$ is homeomorphic to the planar embedding of $G$,
	\item the intersection of each vertex $v$ and each edge $e$ with a cell $\Cell$ of
		$\CurArr$ is non-empty if and only if the intersection of $v$ and $e$
		 with $\phi(\Cell)$ in $(\Gamma, \Arr)$, respectively, is non-empty,
	\item if an edge $uv$ (directed from $u$ to $v$) intersects a sequence of cells $\Cell_1, \Cell_2, \dots, \Cell_r$
		in this order, then $uv$ intersects in $(\Gamma, \Arr)$ the cells
		$\phi(\Cell_1), \phi(\Cell_2), \dots, \phi(\Cell_r)$ in this order, and 
	\item each line $L_i$ intersects in $\Gamma$ the same vertices and edges as
	$\Cur_i$ in $G$, and it does so in the same order.
%\end{compactenum}
\end{inparaenum}
%
%Observe that this definition implies that each vertex that lies in a cell
%$\Cell$ in $\CurArr$, lies in $\phi(\Cell)$ in $A$. Moreover, each line $L_i$
%intersects in $\Gamma$ the same vertices and edges as $\Cur_i$ in $G$, and it
%does so in the same order.
%
We focus on straight-line aligned drawings. For brevity, unless stated
otherwise, the term aligned drawing refers to a straight-line drawing throughout
this paper.

Note that the stretchability of $\CurArr$ is a necessary condition for the
existence of an aligned drawing. Since testing stretchability is
$\cNP$-hard~\cite{mnev1988universality,shor1991stretchability}, we assume that a
geometric realization $A$ of $\CurArr$ is provided.  Line arrangements
of size up to 8 are always stretchable~\cite{GOODMAN1980385}, and only starting
from nine lines non-stretchable arrangements exist; see the Pappus
configuration~\cite{levi1926teilung} in
Fig.~\ref{fig:intro:aligned_drawing:pappus}.
This figure also illustrates an example of an
$8$-aligned graph with a single edge that does not have an aligned drawing.
It is conceivable that in practical applications, e.g., stemming from user
interactions, the number of lines to stretch is small, justifying
the stretchability assumption.

The aligned drawing convention generalizes the problems studied by Da Lozzo et
al. and Biedl et al. who focused on the case of a single line.  We study a
natural extension of their setting and ask for alignment on general line
arrangements.  
%We note that Da Lozzo et al.  and Biedl and et al. focus on
%aligning a set of vertices as large as possible and separation of a given
%bipartition, respectively. Their characterizations in terms of existence of
%pseudolines allows them to abstract from geometry and to construct the
%pseudolines in a purely  combinatorial setting.  In contrast to that in our
%geometric setting, we are given multiple pseudolines as part of the input. In
%our topological setting we consider the alignment of a given set of vertices $S$
%and give a combinatorial description of a single pseudoline passing through $S$.

In addition to the strongly related works mentioned above, there are several
other works that are related to the alignment of vertices in drawings. 
Ravsky and Verbitsky~\cite{10.1007/978-3-642-25870-1_27} used the fact that $2$-trees have a drawing with at least $n/30$ collinear vertices to show that at least $\sqrt{n / 30}$ vertices of a $2$-tree can be fixed to arbitrary positions. 
Dujmovi{\'{c}}~\cite{Dujmovic2015} shows that every
$n$-vertex planar graph $G=(V, E)$ has a planar straight-line drawing such that
$\Omega(\sqrt{n})$ vertices are aligned, and Da Lozzo et al.~\cite{DaLozzo2016}
show that in planar treewidth-3 and planar treewidth-$k$ graphs, one can align
$\Theta(n)$ and $\Omega(k^{2})$ vertices, respectively.
%vertices and
%that in treewidth-$k$ graphs one can align vertices.  
Chaplik et
al.~\cite{Chaplick2016} study the problem of drawing planar graphs such that all
edges can be covered by $k$ lines.
They show that it is $\cNP$-hard to decide
whether such a drawing exists. 
The computational complexity of deciding whether there exists a  drawing
where all vertices lie on $k$ lines is an open problem~\cite{Chaplick2017}.
Drawings of graphs on $n$ lines where a mapping between the vertices and the lines
is provided have been studied by Dujmovi{\'{c}} et
al.~\cite{Dujmovic2011,Dujmovic2013}.

\smallskip
\noindent
\textit{Contribution \& Outline.} 
After introducing notation in Section~\ref{sec:preliminaries}, 
we first study the topological setting
where we are given a planar graph $G$ and a set $S$ of vertices to align in
Section~\ref{sec:complexity}. 
We show that it is $\cNP$-complete to decide whether $S$ is alignable. On the
positive side, we prove that this problem is fixed-parameter tractable (FPT) with respect to~$|S|$. Afterwards, in
Section~\ref{sec:geometry}, we consider the geometric setting where we
seek an aligned drawing of an aligned graph. 
Based on our proof strategy in Section~\ref{sec:proof_strategy},
we strengthen the result of Da Lozzo et al. and Biedl et al. in
Section~\ref{sec:r_aligned}, and show that there exists a 1-aligned drawing of $G$ with a given convex drawing of the outer face.
In Section~\ref{sec:k_aligned} we consider $k$-aligned graphs with a stretchable
pseudoline arrangement, where every edge
$e$ either entirely lies on a pseudoline or intersects at most
one pseudoline, which can either be in the interior or an endpoint of $e$.  We
utilize the result of Section~\ref{sec:r_aligned} to prove that every such $k$-aligned graph has an
aligned drawing, for any value of $k$. In the preliminaries we define the
\emph{alignment complexity} of an aligned graph. It is a triple that indicates
how many intersections an edge has with the pseudoline arrangement depending on the number of
endpoints that lie on a pseudoline. Table~\ref{tab:solved_cases} summarizes
the results of our paper.

\begin{table}
	\centering
	\begin{tabular}{c | c | l}
		\textbf{alignment complexity} & \textbf{$k$} & \textbf{drawable} \\
		\hline
		$(0, \sentinel, \sentinel)$ & $\geq 1$ &
	 	$\checkmark$ -- Planarity\\
		
		$(0, 0, 0)$ &  $\geq 1$ & $\checkmark$ -- Theorem~\ref{lemma:k_aligned:simple} \\
		$(1, 0, \sentinel)$ & $\geq 1$ &
		$\checkmark$ -- Theorem~\ref{lemma:k_aligned:draw_triangulation} \\
		\hline	
		$(1,0, 0)$ & 2 & open -- Fig.~\ref{fig:2_aligned:2_anchored} \\
		\hline
		$(\sentinel, \sentinel, 2)$ & \multirow{3}{*}{$\geq 8$} & \multirow{3}{*
		}{\ding{55}
		 -- Fig.~\ref{fig:intro:aligned_drawing}(c) }\\
		$(\sentinel, 3, \sentinel)$ & \\
		$(4, \sentinel, \sentinel)$ & \\
	\end{tabular}
	
	\caption{Families of aligned graphs that always have an aligned
		drawing  are marked with $\checkmark$. The symbol \ding{55} indicates that for this
		particular class, there is an aligned graph that does not have an aligned
		drawing. }
	\label{tab:solved_cases}
\end{table}

\section{Preliminaries}
\label{sec:preliminaries}

Let $\CurArr$ be a pseudoline arrangement with $k$ pseudolines $\Cur_1,
\dots, \Cur_k$ and $(G, \CurArr)$ be an aligned graph with $n$ vertices. The set of cells in
$\CurArr$ is denoted by $\CurCells(\CurArr)$.  A cell is \emph{empty} if it does
not contain a vertex of $G$. Removing from a pseudoline its intersections with
other pseudolines gives its \emph{pseudosegments}.

Let $G=(V, E)$ be a planar embedded graph with vertex set $V$ and edge set~$E$.
We call $v \in V$ \emph{interior} if $v$ does not lie on the boundary of the
outer face of $G$. An edge $e\in E$ is \emph{interior} if $e$ does not lie
entirely on the boundary of the outer face of $G$. An interior edge is a
\emph{chord} if it connects two vertices on the outer face.  A point $p$ of an
edge $e$ is an \emph{interior} point of $e$ if $p$ is not an endpoint of $e$.  A
\emph{triangulation} is a biconnected planar embedded graph whose inner faces
are all triangles and whose outer face is bounded by a simple cycle.  A
\emph{triangulation} of a graph $G$ is a triangulation that contains $G$ as a
subgraph.  A \emph{$k$-aligned triangulation} of $(G, \CurArr)$ is a $k$-aligned
graph $(G_T, \CurArr)$ with $G_T$ being a triangulation of $G$.  A graph $G'$ is
a \emph{subdivision} of $G$ if $G'$ is obtained by placing \emph{subdivision
vertices} on edges of $G$.  For an abstract graph $G$ and an edge $e$ of $G$ the
graph~$G/e$ is obtained from $G$ by contracting $e$ and merging the resulting
multiple edges and removing self-loops.  Routing the edges incident to $e$ close
to $e$ yields a planar embedding of $G/e$ in case of a planar embedded graph
$G$.  A \emph{$k$-wheel} is a simple cycle $C$ with $k$ vertices on the outer
face and one additional interior vertex that has an edge to each vertex in $C$.
Let $\Gamma$ be a drawing of $G$ and let $C$ be a cycle in $G$. We denote with
$\Gamma[C]$ the drawing of $C$ in $\Gamma$. Let $T$ be a separating triangle in
$G$ and let $V_{\mathrm{in}}$ and $V_{\mathrm{out}}$ be the vertices in the
interior and exterior of $T$, respectively. We refer to the graphs induced by $T
\cup V_{\mathrm{in}}$ and $T \cup V_{\mathrm{out}}$ as the \emph{split
components of $T$} and denote them by $\gin$ and~$\gout$.

A vertex is \emph{$\Cur_i$-aligned} (or simply \emph{aligned} to $\Cur_i$) if it
lies on the pseudoline $\Cur_i$. A vertex that is not aligned is \emph{free}.
An edge $e$ is
\emph{$\Cur_i$-aligned} (or simply \emph{aligned}) if it completely lies on
$\Cur_i$.  Let $E_{\mathrm{aligned}}$ be the set of all aligned edges. An
\emph{intersection vertex} lies on the intersection of two pseudolines $\Cur_i$
and $\Cur_j$. 
A non-aligned edge is \emph{$i$-anchored} ($i=0,1,2$) if $i$ of its endpoints are aligned
to distinct pseudolines. 
An $\Cur$-aligned edge is \emph{$i$-anchored} ($i=0,1,2$) if $i$ of its endpoints are aligned
to distinct pseudolines which are different from $\Cur$. For example, the single
aligned edge in Fig.~\ref{fig:def_x_chromatic_graphs:a} is $1$-anchored.
Let $E_i$ be the set of $i$-anchored edges; note that, the
set of edges is the disjoint union $E_0 \cupdot E_1 \cupdot E_2$.  
An edge $e$ is (\emph{at most}) \emph{$l$-crossed} if  (at most) $l$ distinct
pseudolines intersect $e$ in its interior. 
A $0$-anchored $0$-crossed non-aligned edge is also called \emph{free}.
A non-empty edge set $A \subset E$
is $l$-crossed if $l$ is the smallest number such that every edge in $A$ is at
most $l$-crossed.
%The empty set is $\sentinel$-crossed. %Then $A$ has the
%\emph{alignment } $l$.

\begin{figure}[tb]
	\centering
	\subfloat[$(1, 0, \sentinel)$\label{fig:def_x_chromatic_graphs:a}]{
		\includegraphics[page=1]{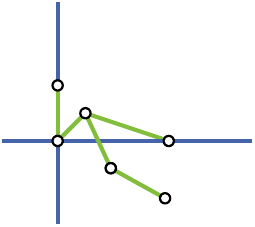}
	}
	\quad
	\subfloat[$(1, 0, 0)$]{
		\includegraphics[page=2]{figures/definition_x-chromatic_graphs}
	}
	\quad
	\subfloat[$(2, 1, 0)$]{
		\includegraphics[page=3]{figures/definition_x-chromatic_graphs}
	}

	\caption{Examples for the alignment complexity of an aligned graph.}
	\label{fig:def_x_chromatic_graphs}
\end{figure}

The \emph{alignment complexity} of an aligned graph 
describes how ``complex'' the relationship between the graph $G$ and the pseudoline
arrangement $ \Cur_1, \dots, \Cur_k$ is. It is formally defined as a triple
$(l_0, l_1, l_2)$, where $l_i$, $i=0,1,2$, indicates that $E_i$ is at most $l_i$-crossed or
has to be empty, if $l_i = \sentinel$.
%if $l_i=\sentintal$ indicates that $E_i$ is
%empty otherwise it the set $E_i$ is $l_i$-crossed.
%indicates that  $E_i$ is $l_i$-crossed.
%where $l_i = \sentinel$ implies that $E_i$ is empty.
For example, an aligned graph where every
vertex is aligned and every edge has at most $l$ interior intersections has the
alignment complexity $(\sentinel, \sentinel, l)$. For further examples, see
Fig.~\ref{fig:def_x_chromatic_graphs}.

\begin{theorem}
	\label{lemma:k_aligned:simple}
		Every $k$-aligned graph $(G, \CurArr)$ of alignment
		complexity $(0, 0, 0)$ with a stretchable pseudoline arrangement $\CurArr$
		has an aligned drawing.
\end{theorem}

\begin{proof}
	We modify the graph $(G, \CurArr)$ as follows; see
	Fig.~\ref{fig:modify_simple_aligned}. We place a vertex
	on each intersection of two or more pseudolines (if the intersection is not
	already occupied). 
	In case that $k$ is at least two,  every unbounded cell $\Cell$ of $\CurArr$ has
	two pseudosegments of infinite length. We place a vertex on each of them at
	infinity and connect them by an edge routed through the interior of $\Cell$.

	\begin{figure}[tb]
		\centering
		\includegraphics{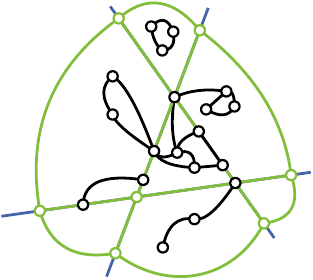}
		\caption{The black edges and vertices and the blue pseudoline arrangement is the
			input graph $(G, \CurArr)$. The green and black graph
			together depict the modified graph before the triangulation step.}
		\label{fig:modify_simple_aligned}
	\end{figure}

	Further, let $u$ and $v$ be two $\Cur$-aligned vertices, that are consecutive
	along $\Cur$. If $uv$ is not already an edge of $G$, we insert it into $G$ and
	route it on $\Cur$. Note that, since $(G, \Cur)$ does not contain edges that
	cross a pseudoline, the resulting graph is again an aligned graph of
	alignment complexity $(0,0,0)$. The boundary of every cell is covered by aligned
	edges. Thus, we can triangulate $(G, \CurArr)$ without introducing
	intersections between edges and a pseudoline.
	
	We obtain an aligned drawing of the modified graph as follows. Note that the
	only interaction between two cells are the aligned vertices and edges on their
	common boundary, i.e., there are no edges crossing the boundary. Hence, for
	every pseudosegments of $\CurArr$ we place the aligned vertices on it,
	arbitrarily (but respecting their order) on the corresponding line segment in
	$\Arr$.	Since, every cell is covered by aligned edges, we can draw the interior
	of two cells independently from each other. More formally, the vertex
	placements of the vertices of the pseudolines prescribes a convex drawing of
	the outer face of the graph $G_{\Cell}$, i.e., the graph induced by the vertices
	in the interior or on the boundary of a cell $\Cell$. Thus, we obtain a
	drawing $\Gamma$ of $G$ by applying the result of
	Tutte~\cite{doi:10.1112/plms/s3-13.1.743} to each graph $G_{\Cell}$,
	independently.
\end{proof}

\section{Complexity and Fixed-Parameter Tractability} \label{sec:complexity} 
In this section, we deal with the topological setting where we are given a planar
embedded graph $G=(V, E)$ and a subset $S \subseteq V$. We ask for a
straight-line drawing of $G$ where the vertices in $S$ are collinear. According
to Da Lozzo et al.~\cite{DaLozzo2016}, this problem is equivalent to deciding
the existence of a pseudoline $\Cur$ with respect to $G$ passing exactly through the vertices
in $S$.  We refer to this problem as \emph{pseudoline existence problem} and the
corresponding search problem is referred to as \emph{pseudoline construction
problem}.
Using techniques similar to Fößmeier and
Kaufmann~\cite{DBLP:conf/ciac/FossmeierK97}, we can show that the pseudoline
existence problem is $\cNP$-hard.

Let $\dual{G} + V$ be the graph obtained from the dual graph $\dual{G} =
(\dual{V}, \dual{E})$ of $G=(V, E)$
by placing every vertex $v \in V$ in its dual face $\dual{v}$ and connecting it
to every vertex on the boundary of the face $\dual{v}$.

\begin{lemma}
	\label{theorem:duality_hamiltonian_pseudoline}
	Let $G=(V,E)$ be a 3-connected 3-regular planar graph. There exists a pseudoline through $V$ with respect to the graph
	$\dual{G}+ V$ if and only if $G$ is Hamiltonian.
\end{lemma}

\begin{proof} Recall that the dual of a $3$-connected
	$3$-regular graph is a triangulation with a single combinatorial embedding.

	Assume that there exists a pseudoline $\Cur$ through $V$
	with respect to $\dual{G}+ V$. Then the order of appearance of the vertices
	of $\dual{G}+ V$ on $\Cur$ defines a sequence of adjacent faces in $\dual{G}$, i.e.,
	vertices of the primal graph~$G$ that are connected via primal edges. This yields a
	Hamiltonian cycle in $G$.

	Let $C$ be a Hamiltonian cycle of $G$ and consider a
	simultaneous embedding of $G$ and $\dual{G} + V$ on the plane, where
	each pair of a primal and its dual edge intersects exactly once.  Thus, the
	cycle $C$ crosses each dual edge $e$ at most once and passes through
	exactly the vertices $V$. There is a vertex $v$ on the cycle $C$ such
	that $v$ lies in the unbounded face of $\dual{G} + V$.  Thus, the cycle $C$
	can be interpreted as a pseudoline $\Cur(V)$ in $\dual{G} + V$ through
	all vertices in $V$  by splitting it in the unbounded face of $\dual{G}
	+ V$.
\end{proof}

Since computing a Hamiltonian cycle in 3-connected 3-regular planar graphs is
$\cNP$-complete~\cite{doi:10.1137/0205049}, we get that the pseudoline
construction problem is
$\cNP$-hard. 
On the other hand, we can guess a sequence of
vertices, edges and faces of $G$, and then test in polynomial time whether this corresponds
to a pseudoline $\Cur$ with respect to $G$ that traverses exactly the vertices
in $S$. Thus, the pseudoline construction problem is in $\cNP$. This proves the
following theorem.

\begin{theorem}
	\label{theorem:hardness_s_aligned}
	The pseudoline existence problem is $\cNP$-complete.
\end{theorem}

\newcommand{\ep}{\mathrm{ep}}
\newcommand{\tr}{\mathrm{tr}}
In the following, we show that the pseudoline construction problem is
fixed-parameter tractable with respect to $|S|$. To this end, we
construct a graph $G^{\tr} = (V^{\tr},E^{\tr})$ and a set
$S^{\tr} \subseteq V^{\tr}$ with $|S^{\tr}| \le |S|+1$ such that
$G^{\tr}$ contains a simple cycle traversing all vertices in
$S^{\tr}$ if and only if there exists a pseudoline $\Cur$ that passes
exactly through the vertices in $S$ such that $(G,\Cur)$ is an aligned
graph.

We observe that if the vertices $S$ of a positive instance are not independent,
they can only induce a \emph{linear forest}, i.e., a set of paths, as
otherwise, there is no pseudoline through all the vertices in $S$ with respect
to $G$.  We call the edges on the induced paths \emph{aligned edges}.  An
edge that is not incident to a vertex in $S$ is called
\emph{crossable}, in the sense that  only crossable edges can be crossed by
$\Cur$, otherwise $\Cur$ is not a pseudoline with respect to $G$.
Let $S_{\ep} \subseteq S$ be the subset of vertices
that are endpoints of the paths induced by $S$ (an isolated vertex is
a path of length~0).  We construct $G^{\tr}$ in several steps; refer to
Fig.~\ref{fig:fpt:transformed_face}.

\begin{figure}[tb]
	\centering
	\includegraphics[page=4]{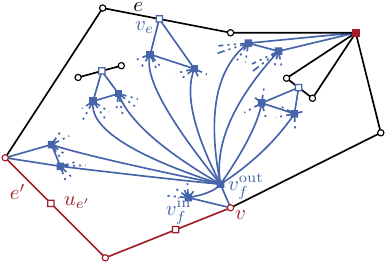}
	\caption{The black and red edges depict a single face of the input
		graph~$G$. Red and blue edges build the transformed graph $G^\tr$. Red round
		vertices are vertices in $S$, red squared
	vertices illustrate the set $S^\tr$, the filled red square is a vertex in $S$
	and $S^\tr$. Blue dashed edges sketch the 
	clique edges between clique vertices (filled blue). }
	\label{fig:fpt:transformed_face}
\end{figure}

\begin{description}
\item[Step~1] Let $G'$ be the graph obtained from $G$ by
  subdividing each aligned edge $e$ with a new vertex $u_{e}$ and let
  $S^{\tr}$ be the set consisting of all isolated vertices in $S$
  and the new subdivision vertices.  Additionally, we add to $G'$
  one new vertex $o$ that we embed in the outer face of $G$ and also
  add to $S^{\tr}$.  Observe that by construction
  $|S^{\tr}| \le |S|+1$.  Finally, subdivide each crossable edge $e$
  by a new vertex $v_{e}$.  
  We call these vertices \emph{traversal nodes} and
  denote their set by $T = S_{\ep} \cup \{v_{e} \mid e$ is
  crossable$\} \cup \{o\}$.  Intuitively, a curve
  will correspond to a path that uses the vertices in $S_{\ep}$ to hop
  onto paths of aligned edges and the subdivision vertices of
  crossable edges to traverse from one face to another. Moreover, the
  vertex $o \in S^\tr$ plays a similar role, forcing the curve to visit the
  outer face.
  
\item[Step~2] For each face $f$ of $G'$ we perform the following
  construction.  Let $T(f)$ denote the traversal nodes that are
  incident to $f$.  For each vertex $v \in T(f)$ we create two new
  vertices $v^{\inp}_{f}$ and $v^{\outp}_{f}$, add the edges $vv^{\inp}_{f}$ and
  $vv^{\outp}_{f}$ to $G'$, and draw them in the interior of $f$.  Finally, we
  create a clique $C(f)$ on the vertex set
  $\{v^{\inp}_{f},v^{\outp}_{f} \mid v \in T(f)\}$, and embed its edges in the
  interior of $f$.
  
\item[Step~3] To obtain $G^{\tr}$ remove all edges of $G'$ that correspond to
	edges of $G$ except those that stem from subdividing an aligned edge of $G$.
\end{description}

\begin{lemma}
  \label{lem:curve-cycle}
  There exists a pseudoline $\Cur$ traversing exactly the vertices in $S$
  such that $(G,\Cur)$ is an aligned graph if and only if there exists
  a simple cycle in $G^{\tr}$ that traverses all vertices in
  $S^{\tr}$.
\end{lemma}

\begin{proof}
  Suppose $C$ is a cycle in $G^{\tr}$ that visits all vertices in
  $S^{\tr}$.  Without loss of generality, we assume that there is no
  face $f$ such that $C$ contains a subpath from $v^{\inp}_{f}$ via
  $v$ to $v^{\outp}_{f}$ (or its reverse) for some vertex
  $v \in T(f) \setminus S_{\ep}$, as otherwise we simply
  shortcut this path by the edge $v^{\inp}_{f}v^{\outp}_{f} \in C(f)$.

  \begin{figure}[t]
	  \centering
	  \subfloat[]{
		  \includegraphics[page=1]{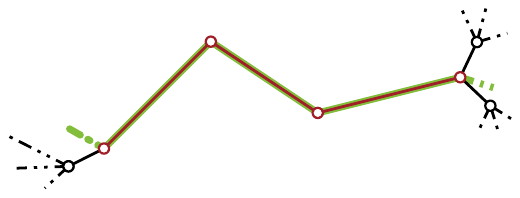}
	  }
	  \quad
	   \subfloat[]{
		  \includegraphics[page=2]{figures/fpt_proof.pdf}
	  }
	  \caption{(a) A pseudoline (thick green) traversing a path of aligned edges
		  (thin red). (b) A
	  path (thick green) in $G^\tr$ visiting consecutive vertices in $S^\tr$ (red squared).}
	  \label{fig:fpt:curve_path}
  \end{figure}

  Consider a path $P$ of aligned edges in $G$ that contains at least
  one edge; refer to Fig.~\ref{fig:fpt:curve_path}.  By definition, $C$ visits all the subdivision vertices
  $u_{e} \in S^{\tr}$ of the edges of $P$, and thus it enters $P$ on
  an endpoint of $P$, traverses $P$ and leaves $P$ at the other
  endpoint.  All isolated vertices of $S$ are contained in
  $S^{\tr}$, and therefore $C$ indeed traverses all vertices in $S$
  (and thus also all aligned edges).  As described above, $G^{\tr}$
  is indeed a topological graph, and thus $C$ corresponds to a closed
  curve $\rho$ that traverses exactly the vertices in $S$ and the aligned
  edges.  
 
  %Consider now the curve $\rho$ as it traverses $G$.
  \begin{figure}[b]
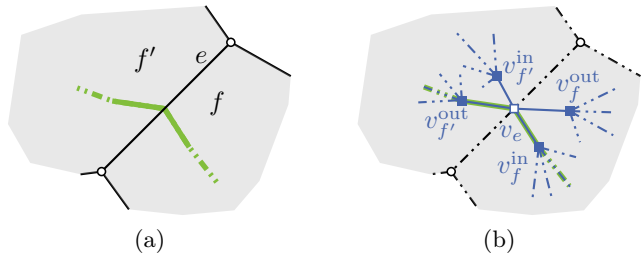

	  \centering
	  \subfloat[]{
		  \includegraphics[page=3]{figures/fpt_proof.pdf}
	  }
	  \qquad
	   \subfloat[]{
		  \includegraphics[page=4]{figures/fpt_proof.pdf}
	  }
	  \caption{(a) A pseudoline (thick green) passing through a non-aligned edge. (b)
	    A path (thick green) in $G^\tr$ traversing a subdivision vertex $v_e$
  		(blue non-filled square). Black (dashed) segments are edges of $G$.}
	  \label{fig:fpt:non_aligned}
  \end{figure}

	We now show that $\rho$ can be transformed to a pseudoline with respect
	to~$G$.  Let $e$ be a non-aligned edge of $G$ that has a common point with
	$\rho$ in its interior; see Fig.~\ref{fig:fpt:non_aligned}.
	%Since $\rho$ intersects $e$,
	Thus, $C$ contains the subdivision vertex $v_{e}$.  In particular, this
	implies that $e$ is crossable.  Moreover, from our assumption on $C$, it
	follows that $C$ enters $v_{e}$ via $v^{\inp}_{f}$ or $v^{\outp}_{f}$ and
	leaves it via $v^{\inp}_{f'}$ or $v^{\outp}_{f'}$, where $f$ and $f'$ are the
	faces incident to $e$, and it is $f \ne f'$ as we could shortcut $C$
	otherwise.  Therefore, $\rho$ indeed intersects $e$ and uses it to traverse to
	a different face of $G$.  Moreover, since $e$ has only a single subdivision
	vertex in $G^{\tr}$ and $C$ is simple, it follows that $e$ is intersected only
	once.  Thus $\rho$ is a curve that intersects all vertices in $S$, traverses
	all aligned edges, and crosses each edge of $G$ (including the endpoints) at
	most once.  Moreover, $\rho$ traverses the outer face since $C$ contains $o$.  
   
  The only reasons why $\rho$ is not necessarily a pseudoline with respect to
  $G$ are that it is a closed curve and it may cross itself.  However, we can
  break $\rho$ in the outer face and route both ends to infinity, and remove
  such self-intersections locally as follows; see
  Fig.~\ref{fig:resolve_crossing}. 
  Consider a circle~$D$ around an intersection~$I$ that neither contains a
  second self-intersection nor a vertex, nor an edge of $G$. Let $\alpha, \beta,
  \gamma, \delta$ be the intersections of $D$ with $\Cur$. We replace the
  pseudosegment $\alpha\gamma$ with a pseudosegment $\alpha \beta$, and
  $\beta\delta$ with a pseudosegment $\gamma\delta$.  We route the
  pseudosegments $\alpha \beta$ and $\gamma\delta$ through the interior of $D$
  such that they do not intersect.  Thus, we obtain a pseudoline $\Cur$ with
  respect to $G$ that contains exactly the vertices in $S$.

\begin{figure}[tb]
	\centering
	\includegraphics{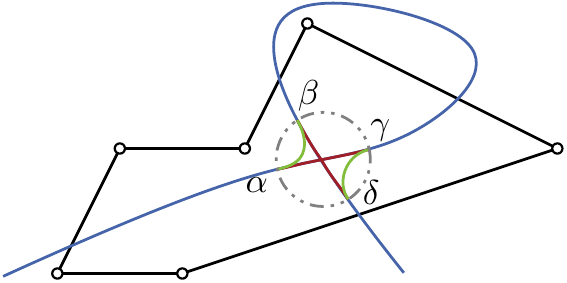}
	\caption{Resolving an intersection by exchanging the intersecting segments (red) with
	non-intersecting segments (green). }
	\label{fig:resolve_crossing}
\end{figure}

  For the converse assume that $\Cur$ is a pseudoline that traverses exactly the
  vertices in $S$ such that $(G,\Cur)$ is an aligned graph. The pseudoline
  $\Cur$ can be split into three parts $\Cur_1$, $\Cur_2$ and $\Cur_3$
  such that $\Cur_1$ and $\Cur_3$ have infinite length and do not intersect
  with $G$, and $\Cur_2$ has its endpoints in the outer face of $G$.  We
  transform $\Cur$ into a closed curve $\Cur'$ by removing $\Cur_1, \Cur_3$
  and adding a new piece connecting the endpoints of $\Cur_2$ without 
  intersecting $\Cur_2$ or $G$. Additionally, we choose an arbitrary
  direction for $\Cur'$ in order to determine an order of the crossed edges and
  vertices.

  We show that $G^{\tr}$ contains a simple cycle traversing the
  vertices in $S^{\tr}$.  By definition $\Cur'$ consists of two
  different types of pieces, see Fig.~\ref{fig:fpt:curve_path}.  The first type traverses a path of
  aligned edges between two vertices in $S_{\ep}$.  The other type
  traverses a face of $G$ by entering and exiting it either via an
  edge or from a vertex in $S_{\ep}$; see Fig.~\ref{fig:fpt:face_taversal}.  We show how to map these pieces
  to paths in $G^{\tr}$; the cycle $C$ is obtained by concatenating
  all these paths.

  \begin{figure}[b]
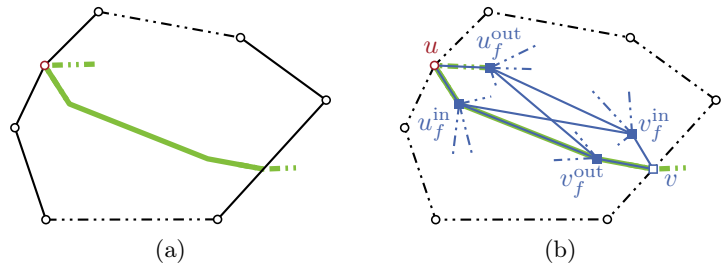

	  \centering
	  \subfloat[]{
		  \includegraphics[page=5]{figures/fpt_proof.pdf}
	  }
	  \qquad
	   \subfloat[]{
		  \includegraphics[page=6]{figures/fpt_proof.pdf}
	  }
	  \caption{(a) A pseudoline piece $\pi$ (thick green) passing through a face $f$.
	  (b) Path (thick green) in $G^\tr$ corresponding to $\pi$.}
	  \label{fig:fpt:face_taversal}
  \end{figure}

	Each piece of the first type indeed corresponds directly to a path in
	$G^{\tr}$; see Fig.~\ref{fig:fpt:curve_path}.  Consider now a piece $\pi$ of
	the second type traversing a face $f$; refer to
	Fig.~\ref{fig:fpt:face_taversal}.  The piece $\pi$ enters $f$ either from a
	vertex in $S_{\ep}$ or by crossing a crossable edge~$e$.  In either case,
	$T(f)$ contains a corresponding traversal node $u$.  Likewise, $T(f)$ contains
	a traversal node $v$ for the edge or vertex that $\Cur'$ intersects next.  We
	map $\pi$ to the path $uu^{\inp}_{f}v^{\outp}_{f}v$ in $G^{\tr}$.  By
	construction, paths corresponding to consecutive pieces of $\Cur'$ share a
	traversal node, and therefore concatenating all paths yields a cycle $C$ in
	$G^{\tr}$.  Moreover, $C$ is simple, since $\Cur'$ intersects each edge and
	each vertex at most once.  Note that $C$ contains at least one edge of the
	outer face (as $\Cur'$ traverses the outer face), and we modify $C$ so
	that it also traverses the special vertex $o$.

  It remains to show that $C$ contains all vertices in $S^{\tr}$.
  There are three types of vertices in $S^{\tr}$; the subdivision
  vertices of aligned edges, the isolated vertices in $S$, and the
  special vertex $o$.  The latter is in $C$ by the last step of the
  construction.  The isolated vertices in $S$ are traversed by $\Cur'$
  and contained in $S_{\ep}$, and they are therefore visited also by
  $C$.  Finally, the subdivision vertices of aligned edges are
  traversed by the paths corresponding to the first type of pieces,
  since $\Cur'$ traverses all aligned edges.
\end{proof}

\begin{theorem}[Wahlstr{\"o}m~\cite{Wahlstrom13}]
\label{theorem:K-cycle}
Given an $n$-vertex graph $G=(V,E)$ and a subset $S\subseteq V$, it can be tested in
$O(2^{|S|}\mathrm{poly}(n))$ time whether a simple cycle through the vertices in
$S$ exists. If affirmative the cycle can be reported within the same asymptotic
time. 
\end{theorem}

\begin{theorem}
  The pseudoline construction problem is solvable in
  $O(2^{|S|}\mathrm{poly}(n))$ time, where $n$ is the number of vertices.
\end{theorem}

\begin{proof}
  Let $G=(V,E)$ with $S \subseteq V$ be an instance of the pseudoline
  construction problem.  By Lemma~\ref{lem:curve-cycle}
  the pseudoline construction problem is equivalent to determining
  whether $G^{\tr}$ contains a simple cycle visiting all vertices in
  $S^{\tr}$.  Since the size of $G^{\tr}$ is $O(n^{2})$ and it can
  be constructed in $O(n^{2})$ time, and $|S^{\tr}| \le |S|+1$,
  Theorem~\ref{theorem:K-cycle} can be used to solve the latter
  problem in the desired running time.
\end{proof}

We note that indeed the construction of $G^{\tr}$ only allows 
leaving a path of aligned edges at an endpoint in $S_{\ep}$.  Therefore,
a single vertex in $S^{\tr}$ for each path of aligned edges would be
sufficient to ensure that $C$ traverses the whole path.  Thus, by
removing for each path all but one vertex from $S^{\tr}$ we
obtain an algorithm that is FPT with respect to the number of paths
induced by $S$.

\begin{theorem} The pseudoline construction problem is solvable in
	$O(2^{P}\mathrm{poly}(n))$ time, where $n$ is the number of vertices and $P$
	is the number of paths induced by the vertex set $S$ to be aligned.
\end{theorem}

\section{Drawing Aligned Graphs}
\label{sec:geometry}
We show that every aligned graph where each edge either entirely lies on a
pseudoline or is intersected by at most one pseudoline, i.e., alignment
complexity $(1, 0, \sentinel)$, has an aligned drawing. For $1$-aligned graphs
we show the stronger statement that every $1$-aligned graph has an aligned
drawing with a given aligned convex drawing of the outer face. We first present
our proof strategy and then deal with $1$- and $k$-aligned graphs.

\subsection{Proof Strategy}
\label{sec:proof_strategy}

Our general strategy for proving the existence of aligned drawings of an aligned
graph $(G, \CurArr)$ is as
follows.
First, we show that we can triangulate $(G, \CurArr)$ by adding vertices
and edges without invalidating its properties.  We can thus assume that our
aligned graph $(G, \CurArr)$ is an aligned triangulation. Second, we show that unless
$G$ has a specific structure (e.g., a $k$-wheel or a triangle), it contains an
aligned or a free edge. Third, we exploit the
existence of such an edge to reduce the instance.
Depending on whether the edge is contained in a
separating triangle or not, we either decompose along that triangle or contract
the edge.  In both cases the problem reduces to smaller instances that are
almost independent. In order to combine solutions, it is, however, crucial to use
the same arrangement of lines $\Arr$ for both of them.

In the following, we introduce the necessary tools used for all three steps
on $k$-aligned graphs of alignment complexity $(1, 0, \sentinel)$. Recall, that
for this class
\begin{inparaenum}[(i)]
	\item every non-aligned edge is at most $1$-crossed,
	\item every $1$-anchored edge is $0$-crossed, and 
	\item there is no edge with its endpoints on two pseudolines.
\end{inparaenum}

Lemmas~\ref{lemma:bi_connected} -- \ref{lemma:k_aligned:proper_triangulation}
show that every aligned graph of alignment complexity $(1, 0, \sentinel)$ has an
aligned triangulation with the same alignment complexity.  If $G$ contains a
separating triangle, Lemma~\ref{lemma:k_aligned:separating_triangle} shows that
$(G, \CurArr)$ admits an aligned drawing if both split components have an
aligned drawing. Finally, with Lemma~\ref{lemma:k_aligned:contraction} we obtain
a drawing of $(G, \CurArr)$ from a drawing of the aligned graph $(G/e, \CurArr)$
where one particular edge $e$ is contracted.

%\begin{lemma}
%  \label{lemma:connecting_vertices}
%  Let $(G, \CurArr)$ be a
%  $k$-aligned graph with two distinct vertices $u, v$ incident to a common
%  face $f$.  Then there is a $k$-aligned graph $(G', \CurArr)$ obtained from
%  $G$ by adding a path $P$ from $u$ to $v$ in the interior of $f$. 
%  The edge set $E(P)$ has alignment complexity $(1, 0, \sentinel)$ and does not
%  contain aligned edges.
%  %The
%  %edges in $E(P)$ are at most $1$-crossed and $0$-anchored, or $0$-crossed and
%  %$1$-anchored.
%  %are $1$-crossed  and at most $1$-anchored. 
%  The size of $P$ is in $O(k^2n)$.
%%}
%\end{lemma}

%\begin{proof}%
%  %
%\begin{figure}[bt]
%  \centering
%  \includegraphics{figures/arr_restricted_to_f}
%  \caption{Connecting two vertices $u, v$ on the boundary of a face (black) with a path
%    (green) in the interior of $f$.}
%  \label{fig:connecting_vertices}
%\end{figure}
%  %
%  We obtain a $k$-aligned graph $(G', \CurArr)$ by routing a path $P$ from $u$
%  to $v$ through the interior of $f$; compare
%  Fig.~\ref{fig:connecting_vertices}.  We construct $P$ such that it has a
%  vertex in every cell of $\CurArr$ it traverses.  
%  As each of the
%  $O(k^2)$ cells in $\CurArr$ may be subdivided into $O(n)$ smaller cells by
%  $f$, a shortest path from $u$ to $v$ contains $O(k^2 n)$ vertices.
%\end{proof}

\begin{lemma}
	\label{lemma:bi_connected}
	Let $(G, \CurArr)$ be a $k$-aligned $n$-vertex graph of alignment complexity $(1,0,
	\sentinel)$.
	Then there exists a biconnected $k$-aligned graph $(G', \CurArr)$ that contains $G$ as a
	subgraph. The set $E(G') \setminus E(G)$ has alignment complexity $(1, 0,
	\sentinel)$ and does not contain aligned edges.
	The size of $E(G') \setminus E(G)$ is in $O(nk + k^3)$.
\end{lemma}

\begin{proof}
	Our procedure works in two steps. First, we connect disconnected components.
	Second, we assure that the graph is biconnected by inserting edges around a
	cut-vertex. Initially, we place a vertex in every cell that does not contain a
	vertex in its interior. 
	
	Consider a cell $\Cell$ of $\CurArr$ that contains two vertices $u$ and $v$
	that belong to distinct connected components $G_u$ and $G_v$.  We refer to two
	vertices $u,v$ that lie in the interior or on the boundary of $\Cell$ as
	\emph{$\Cell$-visible} if there is a curve in the interior of $\Cell$ that
	connects $u$ to $v$ and that does not intersect $G$ except at its endpoints.
	In the following, we exhaustively connect $\Cell$-visible pairs of vertices of
	distinct connected components of $G$.  If $u$ and $v$ are $\Cell$-visible, we
	simply connect them by an edge $e$.  In case that both vertices are aligned,
	we have to subdivide the edge $e$ with a vertex to avoid introducing
	$2$-anchored edges to the graph.  Assume that $u,v$ are not $\Cell$-visible.
	Consider any curve $\rho$ in the interior of $\Cell$ that connects $u$ and
	$v$.
	Then $\rho$ intersects a set of edges of $G$ either in their interior
	or in a vertex. Thus, there are two edges $e_1$ and $e_2$ consecutive along
	$\rho$, that belong two distinct connected components.  Since $e_1$ and $e_2$
	are at most $1$-crossed, there is an endpoint of $e_1$ and an endpoint of $e_2$ that are
	$\Cell$-visible and thus can be connected by an edge.  Overall it is
	sufficient to add a linear number of edges to join distinct connected
	components that have vertices in a common cell.
		
	By construction, every cell contains at least one free vertex.  Thus, in order
	to connect the graph we consider two cells $\Cell_1, \Cell_2$ with a common
	boundary. Assume that there is a vertex $u$ on the common boundary. In this
	case, the previous step ensures that there is a path from $u$ to every vertex
	that lies in the interior or on the boundary of $\Cell_1$ or $\Cell_2$. Hence,
	consider the case where no vertex lies on the common boundary of the two
	cells.  Moreover, the common boundary does also not contain an edge, since
	this edge would be $2$-anchored or $l$-crossed, $l\geq 2$.  Similar to the
	previous step, we can connect two arbitrary vertices of $\Cell_1$ and
	$\Cell_2$ with a curve $\rho$ that intersects the common boundary.  If this
	curve does not intersect an edge we can simply connect the two vertices with
	an edge. Otherwise, at least in one cell $\Cell' \in \{\Cell_1, \Cell_2\}$ the
	curve intersects at least one edge. Therefore, there is an edge $e'$ that
	comes immediately before the intersection of $\rho$ with the boundary of
	$\Cell'$.  Since every edge is at most $1$-crossed, there are two vertices in
	$\Cell_1$ and $\Cell_2$ that can be connected by an edge.  Due to the previous
	step, we can assume that the vertices in the interior of each cell are
	connected by a path.  Thus, we add at most one edge for each pair of adjacent
	cells. 
	%There are at most $k^2$ cells and the dual$	of the pseudoline arrangement is
	%planar. 
	Since there are $O(k^2)$ cells we add $O(k^2)$ vertices and edges to $G$,
	i.e., the size of $G$ is $O(n + k^2)$.

	\begin{figure}[tb]
		\centering
		\subfloat[]{
			\includegraphics[page=1]{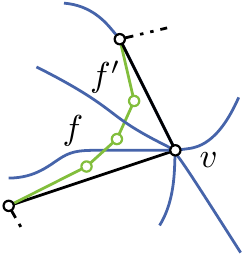}
		}
		\quad
		\subfloat[]{
			\includegraphics[page=3]{figures/biconnected.pdf}
		}

		\caption{ Green edges and vertices are added around a cut-vertex $v$ to connect the
			connected components (black) incident to $v$.
		(a) $v$ is an intersection vertex. (b) $v$ is a free vertex.}
		\label{fig:biconnected}
	\end{figure}

	We now assume that $G$ is connected but not biconnected and has $n' \in
	O(n+k^2)$ vertices.  Consider a single cut vertex $v$; refer to
	Fig.~\ref{fig:biconnected}. We consider the common arrangement $\mathcal F$ of
	$\CurArr$ and $G$, i.e., a face can be restricted by pseudosegments of
	$\CurArr$ and edges of $G$.  Let $\mathcal F_v$ be the set of faces in $\mathcal
	F$ with $v$ on their boundary. We place a vertex $v_f$ in every face $f$ of
	$\mathcal F_v$. Let $f$ and $f'$ be two distinct faces of $\mathcal F_v$ with
	a common edge $\epsilon$ on their boundary.  If $\epsilon$ is an edge $uv$ of
	$G$, we insert the edges $uv_f$ and $uv_{f'}$.  Since $uv$ is at most
	$1$-crossed, the new edges are as well at most $1$-crossed.  If $\epsilon$
	corresponds to a pseudosegment, we insert the edge $v_fv_{f'}$ such that it
	crosses $\epsilon$. Since $v_f$ and $v_{f'}$ are free vertices, the edge is by
	construction $1$-crossed. 
	
	This procedure adds $O(k + \deg v)$
	vertices and edges around $v$, since at most $k$ pseudolines intersect in a
	single point. The degree of vertices adjacent to $v$ is increased by at most
	$2$. Thus, the size of $G$ increases to $O(n'k)$.  Thus, we have that the size
	of $G$ is $O(nk + k^3)$.
\end{proof}

\begin{figure}[tb]
	\centering
	\subfloat[\label{fig:triangulation:intersection}]{
		\includegraphics[page=8]{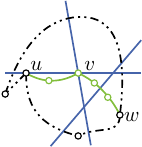}
	}
	\quad
	\subfloat[\label{fig:triangulation:vertex}]{
		\includegraphics[page=10]{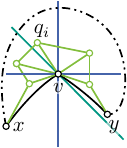}
	}
	\quad
	\subfloat[\label{fig:triangulation:edge}]{
		\includegraphics[page=9]{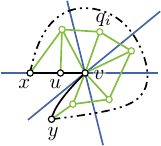}
	}
	\quad
	\subfloat[\label{fig:triangulation:segment}]{
		\includegraphics[page=11]{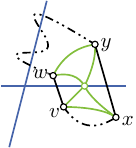}
	}
	\caption{
		Black lines indicate a face $f$ of $G$. Light green edges or vertices are newly
		added into $f$. Blue lines denote the pseudoline arrangement. (a)
		Isolation of an intersection. (b-c) Isolation of an aligned
		vertex or edge. (d) Isolation of a pseudosegment. }
	\label{fig:triangulation}
\end{figure}

\begin{lemma} \label{lemma:k_aligned:triangulation}
	Let $(G, \CurArr)$ be a biconnected $k$-aligned $n$-vertex graph of alignment complexity $(1,
	0, \sentinel)$.	There exists a $k$-aligned triangulation $(G_T=(V_T, E_T),
	\CurArr)$ of $f$ whose size is $O(nk+k^3)$.  The set $E(G_T) \setminus
	E(G)$ has alignment complexity $(1, 0, \sentinel)$ and does not contain
	aligned edges.
\end{lemma}

\begin{proof}%
	We call a face \emph{non-triangular} if its boundary contains more than three
	vertices.
	An aligned vertex $v$ or an aligned edge $e$ is \emph{isolated} if
	all faces with $v$ or $e$ on their boundaries are triangles. A pseudosegment
	$s$ is \emph{isolated} if $s$ does not intersect the interior of a
	simple cycle.
	Our proof distinguishes four cases. Each case is applied exhaustively in
	this order. 
	%Let $f$ be a non-triangular face.
	\begin{compactenum}
	\item If the interior of $f$ contains the intersection of two or more pseudolines, we split the
		face so that there is a vertex that lies on the intersection.
	\item If the
		boundary of a face has an aligned vertex or an aligned edge, we isolate
		the vertex or the edge from $f$.
	\item If the interior of a face
		$f$ intersects a pseudoline $\Cur$, then it subdivides $\Cur$
		into a set of pseudosegments. We isolate each of the pseudosegments
		independently.
  	\item Finally, if none of the previous cases apply,
		i.e., neither the boundary nor the interior of $f$
		contains parts of a pseudoline, the face $f$ can be
		triangulated with a set of additional free edges.
	\end{compactenum}

	Let $\CurArr_f$ be the arrangement of $\CurArr$ restricted to the interior of $f$.

	\begin{compactenum}%
	\item 
		Let $f$ be a non-triangular face whose interior contains an intersection of two
		or more pseudolines; see Fig.~\ref{fig:triangulation:intersection}.
		We
		place a vertex on every intersection in the interior of $f$. We obtain a
		biconnected graph $G_1$ with the application of Lemma~\ref{lemma:bi_connected}.
		Since there are $O(k^2)$ intersections, the size of $G_1$ is $O\left(
		(n+k^2)k + k^3\right) = O(nk + k^3)$.
	\item Let $f_1$ be a non-triangular face of $G_1$ with an aligned vertex or an aligned
		edge $uv$ on its boundary. Further, the interior of $f_1$ does not contain the
		intersection of a set of pseudolines; see
		Fig.~\ref{fig:triangulation:vertex} and \ref{fig:triangulation:edge}. In
		case of an aligned vertex we simply assume $u=v$. Since $G$ is biconnected,
		there exist two edges $xu$, $vy$ on the boundary of $f_1$. 
		 Let $\Cell_1, \dots, \Cell_l \in \CurCells(\CurArr_{f_1})$ be cells with $u$ or
		 $v$ on their boundary, such that $\Cell_i$ is adjacent to $\Cell_{i+1}, i <
		 l$.  Since $f_1$ does not contain $2$-anchored edges, at most one of
		 the vertices $u$ and $v$ can be an intersection vertex. Thus, $l$ is at
		 most $2k$.  We construct an aligned graph $(G_2, \CurArr)$ from $(G_1,
		 \CurArr)$ as follows.  We place a vertex $q_i$ in the interior of each cell
		 $\Cell_i, i\leq l$.  Let $q_0=x$ and $q_{l+1} = y$. We insert edges $e_i =
		 q_iq_{i+1}, i =0,\dots, l$ in the interior of $f_1$ so that the interior
		 of $e_i$ crosses the common boundary of $\Cell_i$ and $\Cell_{i+1}$ exactly
		 once and it crosses no other boundary.  Thus, if the edge $e_i$ is either incident
		 to $x$ or to $y$, it at most $1$-anchored and $0$-crossed. Otherwise, it is
		 $0$-anchored and $1$-crossed. The added path splits $f$ into two faces $f',
		 f''$ with a unique face $f'$ containing $u$ and $v$ on its boundary. If $w
		 \in \{u,v \}$ is on the boundary of cell $\Cell_i$, we insert an edge
		 $wq_i$. Each edge $wq_i$ is $1$-anchored and $0$-crossed.  Let $\Cell_i$
		 and $\Cell_{i+1}$ be two cells incident to $w$. Then, the vertices $w, q_i,
		 q_{i+1}$ form a triangle. If $u \not=v$, there is a unique cell $\Cell_i$
		 incident to $u$ and $v$.  Hence, the vertices $u,v, q_i$ form a triangle.
		 Moreover, for $1 \leq i \leq l$, every edge $uq_i$ and $vq_i$ is incident
		 to two triangles.  Therefore, $f'$ is triangulated. By construction, we do
		 not insert aligned vertices and edges, thus the number of aligned
		 edges and aligned vertices of $f''$ is one less compared to $f_1$. Hence, we
		 can inductively proceed on $f''$.

		Assume the aligned vertex $v$ is an intersection vertex. 
		Thus, isolating $v$ uses $O(k)$ additional vertices and
		edges. Therefore, all intersection vertices can be isolated with $O(k^3)$
		vertices and edges.
		
		Now consider an aligned vertex $v$ that is not an intersection vertex.  In
		this case $v$ is incident to at most two cells.  We can isolate all such
		vertices with $O(n)$ vertices and edges. The same bound holds for aligned
		edges. Finally, we obtain an aligned graph
		$(G_2, \CurArr)$ of size $O(nk +k^3)$. 
			
	\item Let $f_2$ be a non-triangular face of $G_2$ whose interior intersects a
		pseudoline $\Cur$ and has no aligned edge and no aligned vertex on its
		boundary. Further, the interior of $f_2$ does not contain the intersection of
		two or more pseudolines.  Then the face $f_2$ subdivides $\Cur$ into a set of
		pseudosegments; see Fig.~\ref{fig:triangulation:segment}.  We iteratively
		isolate such a pseudosegment $\CurSeg$. 
		%Let $E'$ be the set of edges intersecting a pseudosegment $\CurSeg$ (at the
		%endpoints of $\CurSeg$). 
		Since $f_2$ does not contain the intersection of two or more pseudolines in
		its interior, there are two distinct cells $\Cell_1 \in
		\CurCells(\CurArr_f)$ and $\Cell_2 \in \CurCells(\CurArr_f)$ with $\CurSeg$
		on their boundary.  Since $f_1$ neither contains an aligned vertex nor an
		aligned edge and $G$ is biconnected, there are exactly two edges $e_1=vw$
		and $e_2=xy$ with the endpoints of $\CurSeg$ in the interior of these edges
		and $v,x$  and $w, y$ on the boundaries of $\Cell_1$ and $\Cell_2$,
		respectively.
		%Since $f$ does not contain an intersection of two pseudolines, the vertices
		%$v$ ($w$) and $x$ ($y$) lie on the boundary of the same cell, and $v$ and
		%$w$ on the boundary of different cells. 
		Since $f_2$ does not have an $l$-crossed edge, $l \geq 2$, and every
		$1$-crossed edge is $0$-anchored, the vertices $v$, $w$, $x$, $y$ are free.
		We construct a graph $G'$ by placing a vertex $u$ on $s$ and inserting edges
		$uv$, $uw$, $ux$ $uy$, $vx$ and $wy$.  We route each edge so that the
		interior of an edge does not intersect the boundary of a cell $\Cell_i,
		i=1,2$. Thus, the edges $vx$ and $wy$ are free and the others are
		$1$-anchored and $0$-crossed.

		Every edge in $G_2$ is at most $1$-crossed, thus the number of
		pseudosegments is linear in the size of $G_2$. Therefore, we add a number
		of vertices and edges that is linear in the size of $G_2$.
	
		Thus, we obtain an aligned graph $(G_3, \CurArr)$ of size
		$O(nk + k^3)$.  
\item If none of the cases above applies to a non-triangular face $f_4$ of $G_3$, then neither
	the interior nor the boundary of the face intersects a pseudoline $\Cur_i$.
	Thus, we can triangulate $f_4$ with a number of free edges linear in the size
	of $f_4$.  Thus, in total we obtain an aligned triangulation $(G_T, \CurArr)$
	of $(G, \CurArr)$ of size $O(nk+k^3)$.
\end{compactenum}
\end{proof}

Observe that the correctness of the previous triangulation procedure only relies
on the fact that every non-triangular face contains at most $1$-crossed edges. 
While Lemma~\ref{lemma:k_aligned:triangulation} is sufficient for our purposes,
for the sake of generality, we show how to
isolate $l$-crossed edges. This allows us to triangulate biconnected aligned graphs without increasing the alignment complexity. 

\begin{figure}
	\centering
	%\subfloat[]{
		\includegraphics[page=3]{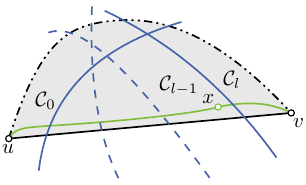}
	%}
	\caption{An $l$-crossed edge $uv$ in a (grey) face $f$ and a pseudoline
		arrangement (blue). The green edges isolate the edge $uv$.}
	\label{fig:isolate_l_crossed_edges}
\end{figure}

\begin{theorem}
	Every biconnected $k$-aligned $n$-vertex graph $(G, \CurArr)$ of alignment
	complexity $(l_0, l_1, l_2)$ has an aligned triangulation $(G_T, \CurArr)$.
	The alignment complexity of $E(G_T) \setminus E(G) $ is $(\max\{l_0, 1 \},
	l_1, l_2)$ and the size of this set is \mbox{$O(nk + k^3)$}.
\end{theorem}

\begin{proof}%
For $l \geq 1$, we iteratively isolate $l$-crossed edges $uv$ from a
non-triangular face~$f$ as sketched in Fig.~\ref{fig:isolate_l_crossed_edges}.
Let $\Cell_0, \Cell_1, \dots, \Cell_l \in \CurCells(\CurArr)$ be the cells in
$f$ that occur in this order along $uv$.  If one of these vertices is free, say
$v$,  we place a new vertex $x$ in the interior of $\Cell_{l - 1}$.  We insert
the two edges $ux, xv$ and route both edges close to $uv$. This isolates the
edge $uv$ from $f$.
Notice that the edge
$xv$ is $0$-anchored and $1$-crossed and the edge $ux$ $(l-1)$-crossed.  In case that $l_0 \geq 1$, the alignment
complexity of the new aligned graph is $(l_0, l_1, l_2)$. Otherwise, the
alignment complexity is $(1, l_1, l_2)$.  If $u$ and $v$ are aligned, we place
$x$ on the boundary of $\Cell_{l-1}$ and $\Cell_{l}$ and route the edges $ux$
and $vx$ as before. The alignment complexity is not affected by this operation.
The face $uvx$ is triangular and therefore the edge $uv$ is processed as above
at most twice.

This procedure introduces a new $(l-1)$-crossed
edge. Repeating the process $l-2$ times generates a new face $f'$ from $f$ where edge
$uv$ is substituted by a path of at most $1$-crossed edges.
To isolate all $l$-crossed edges in $(G, \CurArr)$, we add $O(kn)$ vertices and
edges.

By isolating all $l$-crossed edges in this way,  we obtain an aligned graph
where every non-triangular face is bounded by at most $1$-crossed edges. The
proof of Lemma~\ref{lemma:k_aligned:triangulation} handles all non-triangular
faces independently.  For the correctness of the triangulation it is sufficient
to ensure that every non-triangular face does neither contain $2$-anchored edges
nor $l$-crossed edges. Thus, we can apply the methods used in the proof of
Lemma~\ref{lemma:k_aligned:triangulation} to triangulate $(G, \CurArr)$ with
$O(nk + k^3)$ additional vertices and edges. 
\end{proof}

We now return to the treatment of aligned graphs with alignment complexity $(1,
0, \sentinel)$.
To simplify the proofs, we augment the input graph with an additional
cycle in the outer face that contains all intersections of $\CurArr$
in its interior, and we add subdivision vertices on the intersections of
$\Cur_i$-aligned edges with pseudolines $\Cur_j$, $i \not=j$.
%
%A simple cycle $C$ in $(G, \CurArr)$ is \emph{proper} if and only if
%\begin{inparaenum}[(i)] \item $C$ contains all intersections of $\CurArr$ in
%its interior, and \item and intersects every pseudoline exactly twice.
%\end{inparaenum}
A $k$-aligned graph is \emph{proper} if 
\begin{inparaenum}[(i)]
\item every aligned edge is $0$-crossed,
\item for $k\geq 2$, every edge on the outer face is $1$-crossed,
\item the boundary of the outer face intersects every pseudoline exactly twice, and 
\item the outer face does not contain any intersection of $\CurArr$.
\end{inparaenum}

%if and only if the outer face is proper and every aligned edge is $0$-crossed.
An aligned graph $(G_\mathrm{rs}, \CurArr)$ is a \emph{rigid subdivision} of an aligned
graph $(G, \CurArr)$ if and only if $G_\mathrm{rs}$ is a subdivision of $G$ and every
subdivision vertex is an intersection vertex with respect to $\CurArr$.  We show
that we can extend every $k$-aligned graph $(G, \CurArr)$ to a proper $k$-aligned
triangulation.

\begin{lemma} \label{lemma:k_aligned:proper_triangulation} 
	For every $k\geq 2$ and every $k$-aligned $n$-vertex graph $(G, \CurArr)$ of alignment
	complexity $(1, 0, \sentinel)$, let $(G_\mathrm{rs}, \CurArr)$ be a rigid
	subdivision of $(G, \CurArr)$.  Then there exists a proper $k$-aligned
	triangulation $(G', \CurArr)$ of alignment complexity $(1, 0, \sentinel)$ such
	that $G_\mathrm{rs}$ is a subgraph of $G'$.  The size of $G'$ is in $O(nk^2
	+ k^{4})$. The set $E(G') \setminus E(G_\mathrm{rs})$ has alignment complexity
	$(1, 0, \sentinel)$ and does not contain aligned edges.
\end{lemma}

\begin{proof}
	We construct a rigid subdivision $(G_\mathrm{rs}, \CurArr)$ from $(G,
	\CurArr)$ by placing subdivision vertices on the intersections of
	$\Cur_i$-aligned edges with pseudolines $\Cur_j, i\not=j$.
	The number  $n_{\mathrm{rs}}$ of vertices of $G_\mathrm{rs}$ is in $O(n + k^2)$.

	We obtain a proper biconnected $k$-aligned graph $(G_b, \CurArr)$ by embedding
	a simple cycle $C$ in the outer face of $G_\mathrm{rs}$ and applying
	Lemma~\ref{lemma:bi_connected}. In order to construct $C$, we
	place a vertex $v_c$ in each unbounded cell $c$ of $\CurArr$ and connect two
	vertices $v_c$ and $v_{c'}$ if the boundaries of the cells $c$ and $c'$
	intersect.
	The size $n_b$ of $G_b$ is
	$O(n_\mathrm{rs}k + k^3)  = O(nk + k^3)$.
	We obtain a proper $k$-aligned
	triangulation $(G', \CurArr)$ of $G_b$ with the application of 
	Lemma~\ref{lemma:k_aligned:triangulation}.  The size $n'$ of $G'$ is in
	$O(n_bk + k^3) = O((nk + k^3)k + k^3) = O(nk^2 + k^4)$.
\end{proof}

The following two lemmas show that we can reduce the size of the aligned graph
and obtain a drawing by merging two drawings or by geometrically uncontracting 
an edge.

\newcommand{\lemmaSeparatingTriangle}{
	Let $(G, \CurArr)$ be a $k$-aligned triangulation.
	Let $T$ be a separating triangle splitting $G$
	into subgraphs $\gin, \gout$ so that $\gin \cap \gout=T$ and $\gout$
	contains the outer face of $G$. Then, \begin{inparaenum}[(i)]
		\item $(\gout, \CurArr)$ and $(\gin, \CurArr)$ are $k$-aligned
			triangulations, and
		\item $(G, \CurArr)$ has an aligned drawing if and only if 
			there exists a common line arrangement $A$  such that
			$(\gout, \CurArr)$ has an
			aligned drawing $(\Dout,  \Arr)$ and $(\gin, \CurArr)$ has an aligned
			drawing $(\Din, \Arr)$ with the outer face drawn \mbox{as
			$\Dout[T]$}.
	\end{inparaenum}
}
\begin{lemma}
	\label{lemma:k_aligned:separating_triangle}
	\lemmaSeparatingTriangle
\end{lemma}

\begin{proof}
	It is easy to verify that $(\gout, \CurArr)$ and $(\gin, \CurArr)$ are
	aligned triangulations.
	An aligned drawing $(\Gamma, \Arr)$ of $(G,\CurArr)$
	immediately implies the existence of an aligned drawing
	$(\Dout,\Arr)$ of $(\gout, \CurArr)$ and $(\Din, \Arr)$ of $(\gin, \CurArr)$.

	Let $(\Dout, \Arr)$ be an aligned drawing of
	$(\gout, \CurArr)$.
	Since $(\Dout, \Arr)$ is an aligned drawing, $(\Dout[T], \Arr)$ is an aligned
	drawing of $(T, \CurArr)$. 
	Let $(\Din, \Arr)$ be an aligned drawing of $(\gin, \CurArr)$ with the outer
	face drawn as $\Dout[T]$.
	Let $\Gamma$ be the drawings obtained by merging the drawing
	$\Dout$ and $\Din$. Since $(\Dout, \Arr)$ and $(\Din, \Arr)$ are aligned
	drawings on the same line arrangement $\Arr$, $(\Gamma, \Arr)$ is an aligned drawing of
	$(G, \CurArr)$.
\end{proof}

\newcommand{\lemmaContraction}{
	Let $(G, \CurArr)$ be a
	proper $k$-aligned triangulation of alignment complexity $(1,0,\sentinel)$ and let
	$e$ be an interior $0$-anchored aligned edge or an interior free edge of $G$ that does not
	belong to a separating triangle and is not a chord.
	Then $(G/e, \CurArr)$ is
	a proper $k$-aligned triangulation of alignment complexity $(1,0,\sentinel)$.
	Further, $(G, \CurArr)$ has an aligned drawing if $(G/e, \CurArr)$ has
	an aligned drawing.
}

\begin{lemma} \label{lemma:k_aligned:contraction}
\lemmaContraction
\end{lemma}

\begin{proof} 
	We first prove that $(G/e, \CurArr)$ is a proper $k$-aligned triangulation.
	Consider a topological drawing of the aligned graph $(G, \CurArr)$.	
	Let $c$ be the vertex in $G/e$ obtained from contracting the edge
	$e=uv$.  We place $c$ at the position of $u$. Thus, all the edges incident to
	$u$ keep their topological properties. We route the edges incident to $v$
	close to the edge $uv$ within the cell from which they arrive to $v$ in $(G,
	\CurArr)$.	
	Since $e$ is not an edge of a separating triangle, $G/e$ is simple and triangulated. 

	Consider a free edge $e$.  Observe that the triangular faces incident to $e$
	do not contain an intersection of two pseudolines in their interior, since 
	$(G, \CurArr)$ does not contain  $l$-crossed edges, for $l\geq 2$.  Therefore, $(G/e, \CurArr)$
	is an aligned triangulation.  Since $e$ is not a chord, $(G/e,
	\CurArr)$ is proper.  Further, $u$ and $v$ lie in the interior of the same
	cell, thus, the edges incident to $c$ have the same alignment complexity as in
	$(G, \CurArr)$.
	
	If $e$ is aligned, it is also $0$-crossed, since $(G, \CurArr)$ is proper.
	Since $e$ is also $0$-anchored, the triangles incident to $e$ do not contain an intersection of two 
	pseudolines and therefore $(G/e, \CurArr)$ is a proper aligned triangulation. The routing of the edges
	incident to $c$, as described above, ensures that the alignment complexity
	is $(1, 0 , \sentinel)$.

	%Thus, we argue that the edges $vw$ have
	%the same intersections as the corresponding edges $cw$ in $(G/e, \CurArr)$.

	%If $e$ is free, $u$ is in the same cell as $v$ and thus
	%every edge $vw$ has the same intersections in $(G$, \CurArr) as $cw$ in
	%$(G/e,\CurArr)$.  Thus, let $e$
	%be a $0$-anchored $\Cur$-aligned edge, for a pseudoline $\Cur \in \CurArr$.
	%Since $G$ is proper and $e$ is $0$-anchored, $e=uv$ does not
	%contain the intersection of two pseudolines in its interior. Thus, $c$ is
	%$\Cur$-aligned and lies on the same pseudosegment as $e$. Therefore, every
	%edge $vw$ in $(G, \CurArr)$ has the same intersections with a pseudoline as the
	%corresponding edge $uw$ in $(G/e, \CurArr)$. Hence, $(G/e, \CurArr)$ is
	%$k$-aligned graph with the alignment complexity $(1,0, \sentinel)$.

\begin{figure}[tb] 
	\centering
	\subfloat[]{
		\includegraphics[page=2]{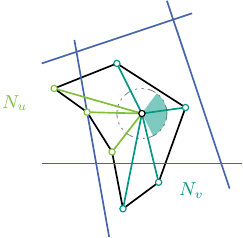} 
	}
	\quad
	\subfloat[]{
		\includegraphics[page=3]{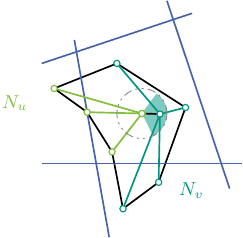} 
	}
	\caption{Unpacking
	an edge in a drawing $\Gamma'$ of $G/e$ (a) to obtain a drawing $\Gamma$
of~$G$ (b). } \label{fig:k_aligned:uncontract_edge}
\end{figure}

	% unpack an edge

	Let $(\Gamma', \Arr)$ be an aligned drawing of $(G/e, \CurArr)$.  We now
	prove that $(G,\CurArr)$ has an aligned drawing.  Let $\Gamma''$ denote the
	drawing obtained from $\Gamma'$ by removing $c$ together with its incident
	edges and let $f$ denote the face of $\Gamma''$ where $c$ used to lie.
	Since $G/e$ is triangulated and $e$ is an interior edge
	and not a chord, $f$ is star-shaped and $c$ lies inside the kernel of $f$; see
	Fig.~\ref{fig:k_aligned:uncontract_edge}.  We construct a drawing $\Gamma$
	of $G$ as follows.  If one of  vertices $u$ and $v$ lies on the outer face, we
	assume, without loss of generality, that vertex to be $u$.  First, we place
	$u$ at the position of $c$ and insert all edges incident to $u$.
	%Note that since $G$ is a proper triangulation and $e$ is $0$-anchored,
	%there is no aligned edge incident to the intersection of two pseudolines or
	%to the outer face.
	This results in a drawing of the face $f'$ in which we have to place $v$.
	Since $u$ is placed in the kernel of $f$, $f'$  is star-shaped.  If $e$ is
	a free edge, the vertex $v$ has to be placed in the same cell as $u$.	We
	then place $v$ inside $f'$ sufficiently close to $c$ so that it lies inside
	the kernel of $f'$ and in the same cell as $u$. All edges incident to $v$
	are at most $1$-crossed, thus, $(\Gamma, \Arr)$ is an aligned drawing of
	$(G, \CurArr)$.  
	
	Likewise, if $e$ is an $\Cur$-aligned edge, then $v$ has to be placed on the
	line $L \in \Arr$ corresponding to $\Cur$. In this case, also  $c$ and therefore
	$u$ lie on $L$. 
	%Since, $G$ is triangulated and
	%and $e=uv$ is $\Cur$-aligned, it
	%is an interior edge of $G$.
	Since $e$ is an interior edge, there exist two triangles $uv, vx, xu$
	and $uv, vy, yu$ sharing the edge $uv$. Since, $e$ is not part of a
	separating triangle, $x$ and $y$ are on different sides of $L$.  Therefore
	the face $f'$ contains a segment of the line $L$ of positive length that is
	within the kernel of $f'$.  Thus, we can place $v$ close to $u$ on the line
	$L$ such that the resulting drawing is an aligned drawing of $(G, \CurArr)$.
\end{proof}

	Note that contracting a $1$-anchored aligned edge can result in a graph $(G/e,
	\CurArr)$ with an alignment complexity that does not coincide with the
	alignment complexity of $(G, \CurArr)$. Further, for general alignment
	complexities there is an aligned graph $(G, \CurArr)$ and an $1$-anchored
	aligned edge $e$ such that $(G/e, \CurArr)$ is not an aligned graph.
	
\subsection{One Pseudoline}
\label{sec:r_aligned}

We show that every $1$-aligned graph $(G, \CurR)$ has an aligned drawing
$(\Gamma, R)$, where $\CurR$ is a single pseudoline and $R$ is the corresponding
straight line. Using the techniques from the previous section, we can assume that $(G,
\CurR)$ is a proper $1$-aligned triangulation. We show that unless $G$ is very
small, it contains an edge with a certain property. This allows for an inductive proof to construct
an aligned drawing of $(G, \CurR)$.

\begin{lemma}
	\label{lem:one-line-triangle}
	Let $(G,\CurR)$ be a proper $1$-aligned triangulation without chords and
	with $k$ vertices on the outer face. If $G$ is neither a
	triangle nor a $k$-wheel whose center is aligned, then $(G, \CurR)$ contains an interior aligned or
	an interior free edge.
\end{lemma}
\begin{proof}
	We first prove two useful claims.

	\smallskip
	\noindent
	\textit{Claim  1.}
		Consider the order in which $\CurR$ intersects the vertices and
		edges of $G$.  If vertices $u$ and $v$ are consecutive on $\CurR$,
		then the edge $uv$ is in $G$ and aligned.
	\smallskip
	
	\begin{claimproof}
		Observe that the edge $uv$ can be inserted into $G$ without creating
		crossings.  Since $G$ is a triangulation, it therefore already contains $uv$, and
		further, since every non-aligned edge has at most one of its endpoints on
		$\CurR$, it follows that indeed $uv$ is aligned.  This proves the claim.
	\end{claimproof}

	\smallskip
	\noindent
	\textit{Claim  2.}
		If $(G,\CurR)$ is an aligned triangulation without aligned edges and
	  $x$ is an interior free vertex of $G$, then $x$ is incident to
		a free edge.

	\smallskip

	\begin{claimproof}
	\begin{figure}[tb]
		\centering
		\subfloat[]{
			\includegraphics[page=2]{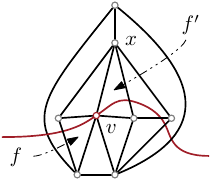}
		}
		\quad
		\subfloat[]{
			\includegraphics[page=3]{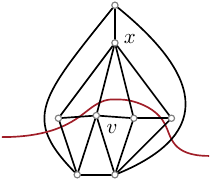}
		}
		\caption{Transformation from a red vertex~(a) to a gray
		vertex~(b).}
		\label{fig:red_to_gray_vertex}
	\end{figure}		
	Assume for a contradiction that all neighbors of $x$ lie either on $\CurR$
	or on the other side of $\CurR$.  First, we slightly modify
	$\CurR$ to a curve $\CurR'$ that does not contain any vertices.
	Assume $v$ is an aligned vertex; see Fig.~\ref{fig:red_to_gray_vertex}.  Since there are no aligned edges, $\CurR$
	enters $v$ from a face $f$ incident to $v$ and leaves it to a
	different face $f'$ incident to $v$.  We then reroute $\CurR$ from
	$f$ to $f'$ locally around $v$.  If $v$ is incident to $x$, we
	choose the rerouting such that it crosses the edge $vx$. 
	
	Notice that if an edge $e$ intersects $\CurR$ in its endpoints, then $\CurR'$
	either does not intersect it or intersects it in an interior point.
	Moreover, $e$ cannot intersect $\CurR'$ twice as in such a case
	$\CurR$ would pass through both its endpoints. 
	%Therefore $(G, \CurR')$ is an
	%aligned graph without any aligned vertices. 
	Now, since $G$ is a triangulation and the outer face of $G$ is proper,
	$\CurR'$ corresponds to a simple
	cycle in the dual $\dual{G}$ of $G$, and hence corresponds to a cut $C$ of
	$G$.  Let $H$ denote the connected component of $G-C$ that contains $x$ and
	note that all edges of $H$ are free.  By the assumption and the construction
	of $\CurR'$, $x$ is the only vertex in $H$.  Thus, $\CurR'$ intersects only
	the faces incident to $x$, which are interior. This contradicts the
	assumption that $\CurR'$ passes through the outer face of $G$ and 
	finishes the proof of the claim.
\end{claimproof}

	We now prove the lemma. Assume that $G$ is neither a triangle nor a
	$k$-wheel whose center is aligned. If $G$ is a $k$-wheel whose center is free,
	we find a free edge by Claim~2.
	Otherwise, $G$ contains at least two interior vertices. If one of
	these vertices is free, we find a free edge by Claim~2. 
	Otherwise, all interior vertices are aligned.
	%therefore the only edge which can intersect $\CurR$ is a chord
	%of $G$.
	Since $G$ does not contain any chord, there is a pair of aligned
	vertices
	consecutive along $\CurR$. Thus by Claim~1 the instance $(G, \CurR)$ has an
	aligned edge.
\end{proof}

\newcommand{\theoremRAlignedDrawing}{
	Let $(G, \CurR)$ be a proper aligned
	graph and let $(\Gamma_O, R)$ be a convex aligned drawing of the
	aligned outer face $(O, \CurR)$ of $G$.
	There exists an aligned drawing $(\Gamma, R)$ of $(G,
	\CurR)$ with the same line $R$ and the outer face drawn as $\Gamma_O$.
}
\begin{theorem}
	\label{theorem:r_aligned:drawing}
	\theoremRAlignedDrawing
\end{theorem}

\begin{proof}
	Given an arbitrary proper aligned graph $(G, \CurR)$, we first complete it to a biconnected graph and then triangulate it by applying
	Lemma~\ref{lemma:bi_connected} and Lemma~\ref{lemma:k_aligned:triangulation}, respectively.

	We prove the claim by induction on the size of $G$.  If $G$ is just a
	triangle, then clearly $(\Gamma_{O},R)$ is the desired drawing.  If $G$ is
	the $k$-wheel whose center is aligned, placing the vertex on the line in the interior of
	$\Gamma_{O}$ yields an aligned drawing of $G$.  This finishes the base case.

	% chords
	If $G$ contains a chord $e$, then $e$ splits $(G, \CurR)$ into two graphs
	$G_1, G_2$ with $G_1 \cap G_2 = e$. It is easy to verify that $(G_i, \CurR)$
	is an aligned graph.  Let $(\Gamma^i_{O}, R)$ be a drawing of the face of
	$\Gamma_O \cup {e}$ whose interior contains $G_i$.  By the inductive
	hypothesis, there exists an aligned drawing of $(\Gamma_i, R)$ with the outer
	face drawn as $(\Gamma_O^i, R)$. We obtain a drawing $\Gamma$ by merging the
	drawings $\Gamma_1$ and $\Gamma_2$. The fact that both $(\Gamma_1, R)$ and
	$(\Gamma_2, R)$ are aligned drawings with a common line $R$ and compatible
	outer faces implies that $(\Gamma,R)$ is an aligned drawing of $(G, \CurR)$.

	%  separating triangle
	If $G$ contains a separating triangle $T$, let $\gin$ and $\gout$ be the
	respective split components with $\gin \cap \gout = T$. By
	Lemma~\ref{lemma:k_aligned:separating_triangle}, the graphs $(\gin, \CurR)$
	and $(\gout, \CurR)$ are aligned graphs. By the induction hypothesis there
	exists an aligned drawing $(\Dout, R)$ of the aligned graphs $(\gout,
	\CurR)$ with the outer face drawn as $(\Gamma_O,R)$.  Let $\Gamma[T]$ be the
	drawing of $T$ in $\Dout$.  Further, $(\gin, \CurR)$ has
	by induction hypothesis an aligned drawing with the outer
	face drawn as $\Gamma[T]$. Thus, by
	Lemma~\ref{lemma:k_aligned:separating_triangle} we obtain an aligned
	drawing of $(G, \CurR)$ with the outer face drawn as $\Gamma_O$.

     If $G$ is neither a
    triangle nor a $k$-wheel, by Lemma~\ref{lem:one-line-triangle}, it contains an interior aligned or
    an interior free edge $e$.    
	%\todo{check paragraph} If $G$ does not contain a separating triangle but a free or an aligned edge $e=uv$ that does not lie on the boundary of the	outer face $O$ and is not a chord of $O$, then consider $(G/e, \CurR)$.
	Since $e$ is not a chord and does not belong to a separating triangle,
	%and $e$ is either an interior aligned edge or an interior free edge, we have 
	by Lemma~\ref{lemma:k_aligned:contraction},
	$(G/e, \CurR)$ is an aligned graph and by the induction hypothesis it has an
	aligned drawing $(\Gamma', R)$ with the outer face drawn as $\Gamma_O$. It
	thus follows by Lemma~\ref{lemma:k_aligned:contraction} again that $(G, \CurR)$
	has an aligned drawing with the outer face drawn as $\Gamma_O$.
\end{proof}

\subsection{Alignment Complexity \texorpdfstring{$(1,0,\sentinel)$}{(1,0,-)}}
\label{sec:k_aligned}

We now consider $k$-aligned graphs $(G, \CurArr)$ of alignment complexity
$(1,0,\sentinel)$, i.e., every edge with two free endpoints intersects at most
one pseudoline, every 1-anchored edge has no interior intersection with a
pseudoline, and $2$-anchored edges are entirely forbidden.  In this section, we
prove that every such $k$-aligned graph has an aligned drawing. As before we can
assume that $(G, \CurArr)$ is a proper aligned triangulation. We show that if
the structure of the graph is not sufficiently simple, it contains an edge with
a special property. Further, we prove that every graph with a sufficiently
simple structure indeed has an aligned drawing. Together this again enables an
inductive proof that $(G, \CurArr)$ has an aligned drawing.
Fig.~\ref{fig:kaligned_base} illustrates the statement of the following lemma.

\begin{figure}[bt]
	\centering
	\includegraphics{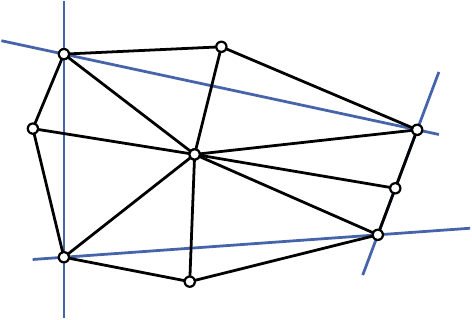}
	\caption{All possible variations of vertices and edges in
		Lemma~\ref{lemma:k_aligned:base_case}.}
	\label{fig:kaligned_base}
\end{figure}

\begin{lemma} \label{lemma:k_aligned:base_case}
	For $k\geq 2$ let $(G, \CurArr)$ be a proper $k$-aligned
	triangulation of alignment complexity $(1, 0, \sentinel)$ that neither
	contains a free edge, nor a $0$-anchored aligned edge, nor a separating
	triangle.
	Then \begin{inparaenum}[(i)]
	\item\label{item:intersection} every intersection contains a vertex,
		%\item every pseudosegment alternately intersects flexible aligned
			%vertices and edges,
		\item\label{item:cell} every cell of the pseudoline arrangement contains exactly one
			free vertex,
		%\item every pseudosegment contains at most one vertex in its interior.
		\item \label{item:ps_cover} every pseudosegment is either covered by two aligned edges or
			it intersects a single~edge.
	\end{inparaenum}
\end{lemma}

\begin{proof}

	The statement follows from the following sequence of claims.
	We refer to an aligned vertex that is not an intersection vertex as a
	\emph{flexible aligned} vertex. 

  \smallskip
	\noindent
	\textit{Claim 1.}
		Every intersection contains a vertex.
	\smallskip

	Assume that there is an intersection $I$ that does not contain a vertex. Since
	$(G, \CurArr)$ is proper,  every aligned edge of $G$ is $0$-crossed. Thus, no
	edge of $G$ contains $I$ in its interior. Moreover, since $(G, \CurArr)$ is a
	proper triangulation, the outer face of $G$ does not contain intersections of
	$\CurArr$. Hence, there is a triangular face $f$ of $G$ that is not the outer
	face and that contains $I$.  Thus, $f$ either has a $2$-anchored edge, a
	$1$-anchored $l_1$-crossed edge, $l_1\geq 1$, or an $l_0$-crossed edge,
	$l_0\geq 2$, on its boundary.  This contradicts that $(G, \CurArr)$ has
	alignment complexity $(1,0, \sentinel)$.

\begin{figure}[b]
%\vspace{-1.5cm}
	\centering
	%\subfloat[\label{fig:k_aligned:free_edges:free_edge}]{
	%  \includegraphics[page=4]{figures/k_aligned_base_case_proof}
	%}
	\quad
	\subfloat[\label{fig:k_aligned:free_edges:cell}]{
		\includegraphics[page=5]{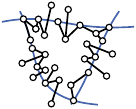}
	}
	\quad
	\subfloat[\label{fig:k_aligned:free_edges:trafo}]{
		\includegraphics[page=6]{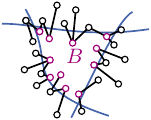}
	}
	\quad
	\subfloat[\label{fig:k_aligned:free_edges:sep_triangle}]{
		\includegraphics[page=7]{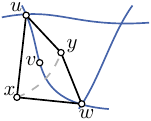}
	}
	\caption{Illustrations for the proof Lemma~\ref{lemma:k_aligned:base_case}.}
	\label{fig:k_aligned:free_edges}
\end{figure}

  	\smallskip%
	\noindent%
	\textit{Claim 2.} %
	Every cell contains at least one free vertex.
	\smallskip%

	\begin{claimproof}	
	Let $\Cell$ be a cell of $\CurArr$. Assume that the boundary of $\Cell$ is
	neither covered by $1$-aligned edges nor crossed by an edge. Since $(G,
	\CurArr)$ is proper, there is a face $f$ of $G$ that entirely contains $\Cell$
	in its interior. Further, $G$ is triangulated and therefore, $f$ is a
	triangle. But every triangle that contains a cell $\Cell$ in its interior
	either has a $2$-anchored edge, a $1$-anchored $l_1$-crossed edge,
	$l_1\geq 1$, or an $l_0$-crossed edge, $l_0\geq 2$, on its boundary. The alignment
	complexity of $(G, \CurArr)$ excludes these types of edges,	thus, there is
	either a $1$-crossed edge with an interior intersection with the boundary of
	$\Cell$, or $\Cell$ is covered by $1$-anchored aligned edges.
	
	If there is an edge $e$ with an interior intersection with the boundary of
	$\Cell$, one endpoint of $e$ lies in the interior of $\Cell$.  Thus, in the
	following we can assume that no such edges exist.  Therefore, the boundary of
	$\Cell$ is covered by $1$-anchored aligned edges.  There are two possibilities
	to triangulate the interior of the cell, either by edges routed through the
	interior of $\Cell$ with endpoints on the boundary of $\Cell$ or with interior
	vertices. The former is not possible, since such a non-aligned edge would
	either be $2$-anchored or have both of its endpoints on the same pseudoline.
	Since $(G, \CurArr)$ is an aligned graph of alignment complexity $(1, 0,
	\sentinel)$, it does not contain such edges.  Thus, every proper aligned
	triangulation of the graph induced by edges on the boundary of $\Cell$
	contains a vertex in the interior of $\Cell$.
	%Thus, the only possibility to triangulate $\Cell$ is with vertices in the
	%interior of $\Cell$. 
	%
\end{claimproof}

	\smallskip
	\noindent
	\textit{Claim 3.}
	Every cell contains at most one free vertex.
	\smallskip
	
	\begin{claimproof} The following proof is similar to Claim~2 in the proof of
		Lemma~\ref{lem:one-line-triangle}.  Let $\Cell$ be a cell and assume for the
		sake of a contradiction that $\Cell$ contains more than one vertex in its
		interior; see Fig.~\ref{fig:k_aligned:free_edges:cell}.  These vertices are
		connected by a set of edges to adjacent cells.  If $\Cell$ contains a vertex $v$ or
		an edge $e$ on its boundary, we reroute the corresponding pseudolines close
		to $v$ and $e$, respectively, such that $v$ and $e$ are now outside of
		$\Cell$; refer to Fig.~\ref{fig:k_aligned:free_edges:trafo}.  Let $\Cell'$
		be the resulting cell, it represents a cut in the graph with two components
		$A$ and $B$, where $\Cell'$ contains $B$ in its interior.  It is not
		difficult to see that the modified pseudolines are still pseudolines with
		respect to $G$.  Since $(G, \CurArr)$ neither contains $2$-anchored edges,
		nor $1$-anchored $l_1$-crossed edges, $l_1 \geq 1$,
		nor $l_0$-crossed edges, $l_0 \geq 2$, every edge of $(G, \CurArr')$ intersects
		the boundary of $\Cell'$ at most once.  Further, $G$ is a triangulation and
		therefore, $B$ is connected and since it contains at least two vertices it
		also contains at least one free edge, contradicting our initial assumption.
	\end{claimproof}
	
	\smallskip
	\noindent
	\textit{Claim  4.}
	Every flexible aligned vertex is incident to two $1$-anchored aligned edges.%
	\smallskip

	\begin{claimproof}
	Let $v$ be a flexible aligned vertex that lies on a pseudosegment $\CurSeg$ of
	$\CurArr$; refer to Fig.~\ref{fig:k_aligned:free_edges:sep_triangle}. Since
	$k\geq 2$, $\CurSeg$ is either incident to one or two intersection vertices.
	Let $u$ be an intersection vertex incident to $\CurSeg$ and let $\CurSeg$ be
	on the boundary of the cells $\Cell_1, \Cell_2$. First, we will show that $u$
	is adjacent to a vertex $x$ in the interior of $\Cell_1$ and a vertex $y$ in
	the interior of $\Cell_2$, respectively.  Depending on whether $\CurSeg$ is
	incident to one or two intersection vertices, the edge $ux$ helps to find
	either a separating triangle or a $4$-cycle that each contains $v$ in its
	interior.

	We initially show that the graph contains the edge $ux$.  Since $G$ is triangulated
	there is a fan of triangles around $u$. Further, all edges in $(G,\CurArr)$
	are at most $1$-crossed, hence we find a vertex $x'$ in the interior of
	$\Cell_1$. Due to Claim~3 and Claim~4 the vertex in the interior of $\Cell_1$
	is unique. Thus, we have that $x'$ is equal to $x$ and therefore $G$ contains
	the edge $ux$. Correspondingly, we find a vertex $y$ in the interior of
	$\Cell_2$ adjacent to $u$.
	
	Consider the case where $\CurSeg$ contains only a single intersection
	vertex, i.e, $\CurSeg$ intersects the outer face of $G$.  Since $(G, \CurArr)$
	is proper (edges on the outer face are $1$-crossed), $G$ contains the edge
	$xy$. Thus, we find a triangle with the vertices $x, y$ and $u$ that contains
	$v$ in its interior.  This contradicts the assumption that $G$ does not have
	a separating triangle.  Therefore, if $\CurSeg$ is incident to a single
	intersection, there is no flexible aligned vertex that lies in the interior of
	$\CurSeg$.
	
	Now consider the case where $\CurSeg$ is incident to two intersection vertices
	$u$ and $w$.
	As shown before, the vertices $u, w$ are each adjacent to the free vertices $x$
	and $y$. Therefore, vertices $u, w, x, y$ build a $4$-cycle containing $v$ in its
	interior.  Since $G$ does not contain a separating triangle, it cannot contain
	the edge $xy$. Moreover, $v$ is the only vertex in the interior of $\CurSeg$,
	as otherwise, we would find a free aligned edge. Finally, since $(G, \CurArr)$ is
	an aligned triangulation, the vertex $v$ is connected to all four vertices
	and thus $v$ is incident to two $1$-anchored aligned edges. 
	\end{claimproof}	

	Claim~1 proves that $(G, \CurArr)$ has Property~(\ref{item:intersection}).
	Claim~2 and Claim~3 together prove that Property~(\ref{item:cell}) is
	satisfied.  Since $(G, \CurArr)$ is an aligned triangulation,
	Property~(\ref{item:ps_cover}) immediately follows from
	Property~(\ref{item:cell}) and Claim 4. 
\end{proof}

\begin{lemma}
	\label{lemma:k_aligned:base_case_drawing}
	Let $(G, \CurArr)$ be a proper $k$-aligned
	triangulation of alignment complexity $(1, 0, \sentinel)$ that does neither
	contain a free edge, nor a $0$-anchored aligned edge, nor a separating
	triangle. Let $\Arr$ be a line arrangement homeomorphic to the pseudoline
	arrangement $\CurArr$.
	Then $(G, \CurArr)$ has an aligned drawing $(\Gamma, \Arr)$.
\end{lemma}

\begin{proof}
	We obtain a drawing $(\Gamma, \Arr)$ by placing every free vertex in its cell,
	every aligned vertex on its pseudosegment and every intersection vertex on its
	intersection.  According to Lemma~\ref{lemma:k_aligned:base_case} every cell
	and every intersection contains exactly one vertex and each pseudosegment is
	either crossed by an edge or it is covered by two aligned edges.  Observe that
	the union of two adjacent cells of the arrangement $\Arr$ is convex.  Thus,
	this drawing of $G$ has an homeomorphic embedding to $(G, \CurArr)$ and every
	edge intersects in $(\Gamma, \Arr)$ the line $L \in \Arr$ corresponding to the
	pseudoline $\Cur \in \CurArr$ in $(G, \CurArr)$
\end{proof}

We prove the following theorem along the same lines as
Theorem~\ref{theorem:r_aligned:drawing}.

\newcommand{\theoremKAlignedDrawing}{
	Every $k$-aligned graph $(G, \CurArr)$ of alignment complexity $(1,
	0,\sentinel)$ with a stretchable pseudoline arrangement $\CurArr$
	has an aligned drawing.
}
\begin{theorem}
	\label{lemma:k_aligned:draw_triangulation}
	\theoremKAlignedDrawing
\end{theorem}

\begin{proof} Let $(G, \CurArr)$ be an arbitrary aligned graph, such that
	$\CurArr$ is a stretchable  pseudoline arrangement, let us denote by $A$ the corresponding line arrangement. By
	Lemma~\ref{lemma:k_aligned:proper_triangulation}, we obtain a proper
	$k$-aligned triangulation $(G_T, \CurArr)$ that contains a rigid subdivision
	of $G$ as a subgraph. Assume that $(G_T, \CurArr)$  has an aligned drawing
	$(\Gamma_T, A)$.  Let $(\Gamma', A)$ be the drawing
	obtained from $(\Gamma_T, A)$ by removing all subdivision vertices $v$ and
	merging the two edges incident to $v$ at the common endpoint. Recall that a
	subdivision vertex in a rigid subdivision of $(G, \CurArr)$ lies on an
	intersection in $\CurArr$.  Hence the drawing $(\Gamma', A)$ is a
	straight-line aligned drawing and contains an aligned drawing $(\Gamma, A)$
	of $(G, \CurArr)$. 
	
	We now show that $(G_T, \CurArr)$ indeed has an aligned drawing.
	We prove this by induction on the size of the instance $(G_T, \CurArr)$.
	If $(G_T, \CurArr)$  neither
	contains a %interior 
	free edge, nor a $0$-anchored aligned edge, nor a separating
	triangle, then, by Lemma~\ref{lemma:k_aligned:base_case_drawing} there is
	an aligned drawing $(\Gamma_T, \Arr)$.

	If $G$ contains a separating triangle $T$, let $\gin$ and $\gout$ be the
	respective split components with $\gin \cap \gout = T$. Since the alignment complexity of  $(G, \CurArr)$  
	is $(1,	0,\sentinel)$,  triangle $T$
	is intersected by at most one pseudoline $\Cur$. It follows that
	$(\gout, \CurArr)$ is a $k$-aligned triangulation and that $(\gin, \Cur)$ is
	a $1$-aligned triangulation.
	By the induction hypothesis there exists an aligned drawing $(\Dout,
	\Arr)$ of  $(\gout, \CurArr)$.  Let $\Dout[T]$ be
	the drawing of $T$ in $\Dout$.  By
	Theorem~\ref{theorem:r_aligned:drawing}, we obtain an aligned drawing
	$(\Din, L)$ with $T$ drawn as $\Dout[T]$.  Moreover, since the drawing
	of $T$ is fixed and is intersected only by line $L$, $(\Din, \Arr)$ is an aligned drawing.  Thus,
	according to Lemma~\ref{lemma:k_aligned:separating_triangle}, there
	exists an aligned drawing of $(G, \CurArr)$.

	 If $G_T$ does not contain separating triangles but contains either a
	free edge or a $0$-anchored aligned edge $e$, let $G_T/e$
	be the graph after the contraction of $e$.  Observe that, since $(G_T,
	\CurArr)$ is proper,  every edge on the outer face is $1$-crossed,
	and therefore every chord is $\ell$-crossed, $\ell\geq1$. Thus, $e$ is an interior edge  
	of $(G_T, \CurArr)$ and is not a chord. Therefore, by Lemma~\ref{lemma:k_aligned:contraction},
	$(G_T/e, \CurArr)$ is a proper aligned triangulation. 
	By induction hypothesis, there exists an aligned
	drawing of $(G_T/e, \CurArr)$, and thus, by
	the same lemma, there exists an aligned drawing of
	$(G_T, \CurArr)$. 
\end{proof}

\begin{figure}[h]
	\centering
	\includegraphics[page=3]{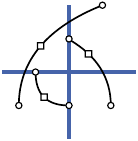}
	\caption{Placement of subdivision vertex to obtain a $2$-aligned graph
of alignment complexity $(1, 0, \sentinel)$.}
	\label{fig:2_aligned:subdivision}
\end{figure}

\begin{theorem}
	\label{theorem:2_aligned_one_bend}
	Every $2$-aligned graph has an aligned drawing with at most one bend per edge.
\end{theorem}
\begin{proof}
We subdivide $2$-crossed, $2$-anchored or $1$-crossed $1$-anchored edges
as depicted in Fig.~\ref{fig:2_aligned:subdivision}. Thus, we obtain a 2-aligned
graph $(G', \CurArr)$ of alignment complexity
$(1,0, \sentinel)$. Applying Theorem~\ref{lemma:k_aligned:draw_triangulation} to $(G',
\CurArr)$ yields a one bend drawing of $(G, \CurArr)$.
\end{proof}

\section{Conclusion}
\label{sec:conclusion}

\begin{figure}[t!]
	\centering	
	\includegraphics[page=2]{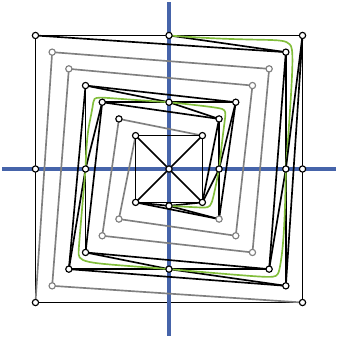}
	\caption{
		 Sketch of a $2$-aligned triangulation without
	aligned or free edges. The green edges are $2$-anchored. 
	The triangulation can be completed as indicated by the
	black edges.}

	\label{fig:2_aligned:2_anchored}
\end{figure}

In this paper, we showed that if $\CurArr$ is stretchable, then every $k$-aligned
graph $(G, \CurArr)$ of alignment complexity $(1, 0,\sentinel)$ has a
straight-line aligned drawing.  As an intermediate result, we showed that a
$1$-aligned graph $(G, \CurR)$ has an aligned drawing with a fixed convex drawing of
the outer face. We showed that the less restricted version of this problem,
where we are only given a set of vertices to be aligned, is $\cNP$-hard but
fixed-parameter tractable.

The case of more general alignment complexities is wide open; refer to
Table~\ref{tab:solved_cases}. Our techniques imply the existence of one-bend
aligned drawings of general 2-aligned graphs as
Theorem~\ref{theorem:2_aligned_one_bend} shows.  However, the existence of
straight-line aligned drawings are unknown even if in addition to $1$-crossed
edges, we only allow $2$-anchored edges, i.e., in the case of alignment
complexity $(1,0,0)$.  In particular, there exist $2$-aligned graphs that
neither contain a free edge nor an aligned edge but their size is unbounded in
the size of the arrangement; see Fig.~\ref{fig:2_aligned:2_anchored}.  It seems
that further reductions are necessary to arrive at a base case that can easily
be drawn.  This motivates the following questions.

\begin{compactenum}[1)]
\item What are all the combinations of line numbers $k$ and alignment
	complexities $C$ such that for every $k$-aligned graph $(G,\CurArr)$ of
	alignment complexity $C$ there exists a straight-line aligned drawing provided
	$\CurArr$ is stretchable?

\item Given a $k$-aligned graph $(G,\CurArr)$ and a line arrangement $A$
	homeomorphic to $\CurArr$, what is the computational complexity of deciding whether $(G,\CurArr)$ admits a straight-line aligned drawing $(\Gamma,A)$?
\end{compactenum}

\paragraph{Acknowledgements}  We thank the anonymous reviewers for thoroughly reading an earlier version of this
paper and for providing useful comments.

\bibliographystyle{abbrvurl}
\bibliography{strings,vol}

\newcommand{\bibdac}[2]{Proceedings of the #1 Annual Design Automation
  Conference (DAC'#2)} \newcommand{\bibinvisau}[1]{Proceedings of the
  Australian Symposium on Information Visualisation (invis.au #1)}
  \newcommand{\bibieeepdp}[2]{Proceedings of the #1 IEEE Symposium on Parallel
  and Distributed Processing #2} \newcommand{\bibieeecs}[1]{Proceedings of the
  IEEE International Symposium on Circuits and Systems #1}
  \newcommand{\bibcccg}[2]{Proceedings of the #1 Canadian Conference on
  Computational Geometry (CCCG'#2)} \newcommand{\bibswat}[2]{Proceedings of the
  #1 Scandinavian Workshop on Algorithm Theory (SWAT'#2)}
  \newcommand{\bibipco}[2]{Proceedings of the #1 International Conference on
  Integer Programming and Combinatorial Optimization (IPCO'#2)}
  \newcommand{\bibsofsem}[2]{Proceedings of the #1 Conference on Current Trends
  in Theory and Practice of Computer Science (SOFSEM'#2)}
  \newcommand{\bibstoc}[2]{Proceedings of the #1 Annual ACM Symposium on Theory
  of Computing (STOC'#2)} \newcommand{\bibfocs}[2]{Proceedings of the #1 Annual
  Symposium on Foundations of Computer Science (FOCS'#2)}
  \newcommand{\bibsoda}[2]{Proceedings of the #1 Annual ACM-SIAM Symposium on
  Discrete Algorithms (SODA'#2)} \newcommand{\bibgd}[2]{Proceedings of the #1
  International Symposium on Graph Drawing (GD'#2)}
  \newcommand{\bibinfovis}[1]{Proceedings of the IEEE Symposium on Information
  Visualization (InfoVis'#1)} \newcommand{\bibvis}[1]{Proceedings of the IEEE
  Conference on Visualization (Vis'#1)} \newcommand{\bibpvis}[1]{Proceedings of
  the IEEE Pacific Visualisation Symposium (PacificVis'#1)}
  \newcommand{\bibsoftvis}[2]{Proceedings of the #1 ACM Symposium on Software
  Visualization (SoftVis'#2)} \newcommand{\bibeurocg}[2]{Proceedings of the #1
  European Workshop on Computational Geometry (EuroCG'#2)}
  \newcommand{\bibsocg}[2]{Proceedings of the #1 Annual Symposium on
  Computational Geometry (SoCG'#2)} \newcommand{\bibwads}[2]{Proceedings of the
  #1 International Symposium on Algorithms and Data Structures (WADS'#2)}
  \newcommand{\bibwg}[2]{Proceedings of the #1 Workshop on Graph-Theoretic
  Concepts in Computer Science (WG'#2)} \newcommand{\bibgta}{Proceedings of the
  Conference at Graph Theory and Applications}
  \newcommand{\bibisaac}[2]{Proceedings of the #1 International Symposium on
  Algorithms and Computation (ISAAC'#2)} \newcommand{\bibcocoon}[2]{Proceedings
  of the #1 Annual International Conference on Computing and Combinatorics
  (COCOON'#2)} \newcommand{\bibtamc}[2]{Proceedings of the #1 Annual Conference
  on Theory and Applications of Models of Computation (TAMC'#2)}
  \newcommand{\bibicalp}[2]{Proceedings of the #1 International Colloquium on
  Automata, Languages and Programming (ICALP'#2)}
  \newcommand{\biblatin}[2]{Proceedings of the #1 Latin American Symposium
  (LATIN'#2)} \newcommand{\bibesa}[2]{Proceedings of the #1 Annual European
  Symposium on Algorithms (ESA'#2)} \newcommand{\bibstacs}[2]{Proceedings of
  the #1 International Symposium on Theoretical Aspects of Computer Science
  (STACS '#2)}
\begin{thebibliography}{10}

\bibitem{Biedl1998}
T.~C. Biedl, M.~Kaufmann, and P.~Mutzel.
\newblock {Drawing Planar Partitions {II}: {HH}-Drawings}.
\newblock In J.~Hromkovi{\v{c}} and O.~S{\'y}kora, editors, {\em
  \bibwg{24th}{98}}, pages 124--136. Springer Berlin/Heidelberg, 1998.
\newblock \href {http://dx.doi.org/10.1007/10692760_11}
  {\path{doi:10.1007/10692760_11}}.

\bibitem{doi:10.1137/130924172}
J.~Cano, C.~D. Tóth, and J.~Urrutia.
\newblock {Upper Bound Constructions for Untangling Planar Geometric Graphs}.
\newblock {\em SIAM Journal on Discrete Mathematics}, 28(4):1935--1943, 2014.
\newblock \href {http://dx.doi.org/10.1137/130924172}
  {\path{doi:10.1137/130924172}}.

\bibitem{Chaplick2016}
S.~Chaplick, K.~Fleszar, F.~Lipp, A.~Ravsky, O.~Verbitsky, and A.~Wolff.
\newblock {Drawing Graphs on Few Lines and Few Planes}.
\newblock In Y.~Hu and M.~N{\"o}llenburg, editors, {\em \bibgd{24th}{16}},
  pages 166--180. Springer-Verlag Springer International Publishing, 2016.
\newblock \href {http://dx.doi.org/10.1007/978-3-319-50106-2_14}
  {\path{doi:10.1007/978-3-319-50106-2_14}}.

\bibitem{Chaplick2017}
S.~Chaplick, K.~Fleszar, F.~Lipp, A.~Ravsky, O.~Verbitsky, and A.~Wolff.
\newblock {The Complexity of Drawing Graphs on Few Lines and Few Planes}.
\newblock In F.~Ellen, A.~Kolokolova, and J.-R. Sack, editors, {\em
  \bibwads{15th}{17}}, pages 265--276. Springer-Verlag Springer International
  Publishing, 2017.
\newblock \href {http://dx.doi.org/10.1007/978-3-319-62127-2_23}
  {\path{doi:10.1007/978-3-319-62127-2_23}}.

\bibitem{DaLozzo2016}
G.~{Da Lozzo}, V.~Dujmovi{\'{c}}, F.~Frati, T.~Mchedlidze, and V.~Roselli.
\newblock {Drawing Planar Graphs with Many Collinear Vertices}.
\newblock In Y.~Hu and M.~N{\"o}llenburg, editors, {\em \bibgd{24th}{16}},
  pages 152--165. Springer-Verlag Springer International Publishing, 2016.
\newblock \href {http://dx.doi.org/10.1007/978-3-319-50106-2_13}
  {\path{doi:10.1007/978-3-319-50106-2_13}}.

\bibitem{Dujmovic2011}
V.~Dujmovi{\'{c}}, W.~Evans, S.~Kobourov, G.~Liotta, C.~Weibel, and S.~Wismath.
\newblock {On Graphs Supported by Line Sets}.
\newblock In U.~Brandes and S.~Cornelsen, editors, {\em \bibgd{18th}{10}},
  pages 177--182. Springer Berlin/Heidelberg, 2011.
\newblock \href {http://dx.doi.org/10.1007/978-3-642-18469-7_16}
  {\path{doi:10.1007/978-3-642-18469-7_16}}.

\bibitem{Dujmovic2013}
V.~Dujmovi{\'{c}} and S.~Langerman.
\newblock {A Center Transversal Theorem for Hyperplanes and Applications to
  Graph Drawing}.
\newblock {\em Discrete \& Computational Geometry}, 49(1), 2013.
\newblock \href {http://dx.doi.org/10.1007/s00454-012-9464-y}
  {\path{doi:10.1007/s00454-012-9464-y}}.

\bibitem{Dujmovic2015}
V.~{Dujmović}.
\newblock {The Utility of Untangling}.
\newblock {\em Journal of Graph Algorithms and Applications}, 21(1):121--134,
  2017.
\newblock \href {http://dx.doi.org/10.7155/jgaa.00407}
  {\path{doi:10.7155/jgaa.00407}}.

\bibitem{DBLP:conf/ciac/FossmeierK97}
U.~F{\"{o}}{\ss}meier and M.~Kaufmann.
\newblock {Nice Drawings for Planar Bipartite Graphs}.
\newblock In G.~Bongiovanni, D.~P. Bovet, and G.~D. Battista, editors, {\em
  {Algorithms and Complexity}}, pages 122--134. Springer Berlin/Heidelberg,
  1997.
\newblock \href {http://dx.doi.org/10.1007/3-540-62592-5_66}
  {\path{doi:10.1007/3-540-62592-5_66}}.

\bibitem{doi:10.1137/0205049}
M.~R. Garey, D.~S. Johnson, and R.~Tarjan.
\newblock {The Planar {Hamiltonian} Circuit Problem is {NP}-Complete}.
\newblock {\em SIAM Journal on Computing}, 5(4):704--714, 1976.
\newblock \href {http://dx.doi.org/10.1137/0205049}
  {\path{doi:10.1137/0205049}}.

\bibitem{GOODMAN1980385}
J.~E. Goodman and R.~Pollack.
\newblock {Proof of {Grünbaum}'s Conjecture on the Stretchability of Certain
  Arrangements of Pseudolines}.
\newblock {\em Journal of Combinatorial Theory, Series A}, 29(3):385 -- 390,
  1980.
\newblock \href {http://dx.doi.org/10.1016/0097-3165(80)90038-2}
  {\path{doi:10.1016/0097-3165(80)90038-2}}.

\bibitem{levi1926teilung}
F.~Levi.
\newblock Die {Teilung} der projektiven {Ebene} durch {Geraden} oder
  {Pseudogerade}.
\newblock {\em {Berichte Mathematische-Physikalischer Klassen der Sächsischen
  Akademie der Wissenschaften Leipzig}}, 78:256--267, 1926.

\bibitem{mnev1988universality}
N.~E. Mn{\"e}v.
\newblock {The Universality Theorems on the Classification Problem of
  Configuration Varieties and Convex Polytopes Varieties}.
\newblock In O.~Y. Viro and A.~M. Vershik, editors, {\em Topology and Geometry
  --- Rohlin Seminar}, volume 1346 of {\em Lecture Notes in Computer Science},
  pages 527--543. Springer Berlin/Heidelberg, 1988.
\newblock \href {http://dx.doi.org/10.1007/BFb0082792}
  {\path{doi:10.1007/BFb0082792}}.

\bibitem{10.1007/978-3-642-25870-1_27}
A.~Ravsky and O.~Verbitsky.
\newblock On collinear sets in straight-line drawings.
\newblock In P.~Kolman and J.~Kratochv{\'i}l, editors, {\em \bibwg{37th}{11}},
  pages 295--306, 2011.

\bibitem{shor1991stretchability}
P.~Shor.
\newblock {Stretchability of Pseudolines is {NP}-Hard}.
\newblock In {\em {Applied Geometry and Discrete Mathematics--The Victor Klee
  Festschrift}}, volume~4 of {\em {DIMACS Series in Discrete Mathematics and
  Theoretical Computer Science}}, pages 531--554. ACM Press, 1991.

\bibitem{doi:10.1112/plms/s3-13.1.743}
W.~T. Tutte.
\newblock {How to Draw a Graph}.
\newblock {\em Proceedings of the London Mathematical Society}, 3(1):743--767,
  1963.
\newblock \href {http://dx.doi.org/10.1112/plms/s3-13.1.743}
  {\path{doi:10.1112/plms/s3-13.1.743}}.

\bibitem{Wahlstrom13}
M.~Wahlstr{\"{o}}m.
\newblock {Abusing the {Tutte} Matrix: An Algebraic Instance Compression for
  the {K-set-cycle} Problem}.
\newblock In N.~Portier and T.~Wilke, editors, {\em \bibstacs{30th}{13}}, pages
  341--352. Schloss Dagstuhl--Leibniz-Zentrum fuer Informatik, 2013.
\newblock \href {http://dx.doi.org/10.4230/LIPIcs.STACS.2013.341}
  {\path{doi:10.4230/LIPIcs.STACS.2013.341}}.

\end{thebibliography}

\end{document}